\documentclass{llncs}

\usepackage{amsmath}
 \usepackage{amssymb}
 \usepackage{float}

 \usepackage{geometry}
\usepackage{color}
\usepackage{ulem} 

\usepackage{mathrsfs}
\usepackage{algorithmicx}
 \usepackage{algorithm}
 \usepackage{algpseudocode}
\usepackage{microtype}
\usepackage{graphicx}
\usepackage{adjustbox}
\newtheorem{observation}{Observation} 
\spnewtheorem{clm}{Claim}{\bfseries}{\rmfamily}
\newcommand{\evc}[0]{\operatorname{evc}}

\newcommand{\mvc}[0]{\operatorname{mvc}}

\title{Eternal vertex cover number of maximal outerplanar graphs}
\titlerunning{Eternal vertex cover number of maximal outerplanar graphs} 

\author{Jasine Babu\inst{1} \and K. Murali Krishnan \inst{2} \and Veena Prabhakaran \inst{1} \and Nandini J. Warrier \inst{2}}
\institute{Indian Institute of Technology Palakkad, India \email{jasine@iitpkd.ac.in,veenaprabhakaran7@gmail.com}
\and National Institute of Technology Calicut, India \email{kmurali@nitc.ac.in, nandini.wj@gmail.com}}

\authorrunning{J.\,Babu, K.\,Murali Krishnan, V.\.Prabhakaran, D.\.Rajendraprasad  and Nandini.\.J.\.Warrier} 

\begin{document}
\maketitle
\begin{abstract}
Eternal vertex cover problem is a variant of the classical vertex cover problem modeled as a two player attacker-defender game. Computing eternal vertex cover number of graphs is known to be NP-hard in general and the complexity status of the problem for bipartite graphs is open. There is a quadratic complexity algorithm known for this problem for chordal graphs. Maximal outerplanar graphs forms a subclass of chordal graphs, for which no algorithm of sub-quadratic time complexity is known. In this paper, we obtain a recursive algorithm of linear time for computing eternal vertex cover number of maximal outerplanar graphs.   
\end{abstract}
\section{Introduction}
A set $S$ of vertices in a graph $G$ forms a vertex cover of $G$ if
every edge in $G$ has at least one end point in $S$.   The size of a minimum vertex cover in $G$, called the vertex cover number of $G$,  will be denoted by $\mvc(G)$. The eternal vertex cover problem of graphs is motivated by the following dynamic network security / fault-tolerance model. A network can be modeled by a graph $G(V,E)$, where the nodes of the network are represented by the vertices of $G$ and the links by the edges of $G$.  The problem is to deploy a minimum set of guards at the nodes, so that if there is an attack (or fault) on a single link at any time, a guard is available at the end of the link, who can move across the link to defend (or repair) the attack (or fault). Simultaneously, the remaining guards need to reconfigure themselves, possibly by repositioning themselves to one of their adjacent nodes, so that any attack (or fault) on a single edge at any future instant of time can also be protected in the same manner. Thus, the model requires guaranteeing protection against single link attacks/failures ad-infinitum.  It is immediate that lowest number of guards needed to achieve this goal in a graph $G$ is at least $\mvc(G)$.  If $k$ guards are sufficient to achieve the goal, then we say that $G$ is \textit{$k$-defendable}.

Formally, guards are initially placed on a vertex cover $C_0$ of $G$, forming an initial configuration.  Suppose at instant $i$, guards are placed in a configuration $C_i$ and an attack occurs on an arbitrary edge $uv\in E(G)$, then a guard at either $u$ or $v$ (or both) must move across the edge.  Other guards may or not move across one of their neighboring edges simultaneously. At the end of these movements, the next configuration $C_{i+1}$ is reached.  To ensure that all edges remain guarded at the instant $i+1$, we require $C_{i+1}$ to be a vertex cover of $G$.  The minimum number of guards necessary for a graph $G$ to be protected against any infinite sequence of single edge attacks is called the eternal vertex cover number of $G$, denoted by $\evc(G)$. The eternal vertex cover problem has two models. The first one allows only at most one guard on a vertex in any configuration, while in the second model this constraint is absent. The results in this paper work in both the models. 

Computing $\evc(G)$ for an arbitrary graph $G$ is NP-hard, though in PSPACE \cite{Fomin2010}, and   is fixed parameter tractable with $\evc(G)$ as parameter\cite{Fomin2010}.  The problem appears significantly harder than $\mvc$ computation.  For instance, $\mvc$ computation is polynomial time for bipartite graphs,  but the complexity status of the problem is open for bipartite graphs\cite{Fomin2010}.   
Consequent to the hardness results, work on solving the problem on various graph classes have been attempted in the literature.  It is known that the problem is NP-complete even for biconnected internally triangulated planar graphs\cite{BABU2021}. Polynomial time algorithms for computing $\evc(G)$ were known exactly only for very elementary graph classes such as an $O(n)$ algorithm for trees \cite{Klostermeyer2009}, a polynomial time algorithm for a tree-like graph class \cite{GeneralizedTrees} and a linear time algorithm for cactus graphs \cite{BPS2021}.   
A recent structure theorem developed in \cite{BP2021} has resulted in a quadratic time algorithm for chordal graphs. Similar graph protection problems defined on parameters like dominating sets~\cite{rinemberg2019,GOLDWASSER2008,KLOSTERMEYER2012,goddard2005,Klostermeyer2011,Anderson2014} and  independent sets~\cite{Hartnell2014,Caro2016EternalIS}  are also studied in literature.

An outerplanar graph is a planar graph that admits a planar embedding with all its vertices lying on the exterior face and it is maximal outerplanar if addition of any more edges between existing vertices will make the graph not outerplanar. Since maximal outerplanar graphs are chordal, it follows from \cite{BP2021} that its $\evc$ number can be computed in quadratic time.  
We improve the complexity and show that the $\evc$ number of maximal outerplanar graphs can be computed in linear time.  

Our algorithm takes a maximal outerplanar graph $G$ on at least three vertices and an edge $uv$ on the outer face of $G$ and recursively computes $\evc(G)$, $\mvc(G)$ along with some related parameters. If both $u$ and $v$ are of degree greater than two, then the recursion is applied on the two induced subgraphs $G_u$ and $G_v$ of $G$ as shown in Fig.~\ref{fig:vc_computation_b}. Otherwise, the recursion works on the graph obtained by deleting the degree two end point of the edge $uv$ from $G$. Since $\evc(G)$ happens to be not merely a function of the eternal vertex cover number and vertex cover number of these associated subgraphs, the algorithm has to recursively compute this larger set of parameters. The linear time complexity of the algorithm is derived from Theorems \ref{thm:vc_evc_deg2-1} to \ref{thm:vc_evc_general-2}, which assert that $\evc(G)$  can be determined in constant time from this larger set of parameters of the associated subgraphs. Apart from the algorithmic result, these theorems serve to demonstrate the structural connection between the $MVC$ problem and the $EVC$ problem for maximal outerplanar graphs.
\section{Preliminaries}
Let $G$ be a graph with $S \subseteq V$. The parameter $\evc_S(G)$ denotes the minimum number $k$
such that $G$ is $k$-defendable with all vertices of
$S$ occupied in all configurations. Similarly,
$\mvc_S(G)$ denote 
the size of the smallest cardinality vertex cover of 
$G$ containing all vertices of $S$. When 
$S=\{v_1 \cdots v_i\}$ for $1\le i \le 3$, we shorten the notation $\evc_{\{v_1 \cdots v_i\}}(G)$ and $\mvc_{\{v_1 \cdots v_i\}}(G)$ as $\evc_{v_1 \cdots v_i}(G)$ and $\mvc_{v_1 \cdots v_i}(G)$ respectively.\\
Proposition~\ref{prop:evc-subset} is an easy adaptation of a result from known literature.
\begin{proposition}\label{prop:evc-subset} \cite{BABU2021}
Let $G(V, E)$ be a maximal outerplanar graph with at least two vertices and $S \subseteq V$. 
 If for every vertex $v \in V \setminus S$,
 $\mvc_{_{S \cup \{v\}}}(G)=\mvc_{_S}(G)$, then $\evc_{_{S}}(G)=\mvc_{_{S}}(G)$. Otherwise, 
 $\evc_{_{S}}(G)=\mvc_{_{S}}(G)+1$. Hence,  $\evc_{_{S}}(G)=\max_{v \in V(G)}\mvc_{_{S \cup \{v\}}}(G)$.
\end{proposition}
\begin{proposition}\label{prop:occ_edge_vertices}
Let $G$ be a maximal outerplanar graph with an edge $uv$.  Then, 
$\evc_v(G) \le \evc_{uv}(G) \le \evc(G)+1$.
\end{proposition}
\begin{proof}
Let $\evc(G)=k$. The first part of the inequality is 
obvious. So, we prove the second part. By 
Proposition~\ref{prop:evc-subset}, for
each vertex $x \in V(G)$, $\mvc_x(G) \le k$.
Since $uv \in E(G)$, this implies that $\mvc_{xu}(G)\le k$ or 
$\mvc_{xv}(G)\le k$. Therefore, for any vertex 
$x \in V(G)$, we have $\mvc_{xuv}(G) \le k+1$. Hence by 
Proposition~\ref{prop:evc-subset}, $\evc_{uv}(G) \le k+1= \evc(G)+1$.\qed
\end{proof}
From now on, whenever we say a graph is maximal outerplanar, we assume a fixed outerplanar embedding of the graph.

For a maximal outerplanar graph $G$ on at least three vertices and any edge $uv$ on the outer face of $G$, we use the notation $\Delta(uv)$ to denote the unique common neighbor of $u$ and $v$.

  \begin{definition}[$uv$-segments]\label{def:uwsegment}
Let  $G$ be a maximal outerplanar graph with at least three vertices.  Let $uv$ be an edge on the outer face of $G$. Let $w=\Delta(uv)$. We define the graph $G_{u}(uv)$ (respectively, $G_{v}(uv)$)  to be the maximal biconnected outerplanar subgraph of $G$ that satisfies the following two properties (see Fig.~\ref{fig:vc_computation_b}):
\begin{enumerate}
     \item The edge $uw$ (respectively, $vw$) is on the outer face of $G_u(uv)$ (respectively, $G_v(uv)$).
     \item   $G_u(uv)$ (respectively, $G_v(uv)$) does not contain the vertex $v$ (respectively $u$). 
\end{enumerate}
$G_u(uv)$ and $G_v(uv) $ will be called the $uv$ segments of $G$.  
\end{definition}  
\begin{note}
  We will write $G_{u}$ and $G_{v}$  instead of $G_u(uv)$ and $G_v(uv)$ when there is no scope for confusion. Observe that $G_u$ (respectively, $G_v$) will be a single edge, if $deg_G(u)=2$ (respectively, $deg_G(v)=2$).
\end{note} 
\begin{figure}
 \centering
 \includegraphics[scale=0.7]{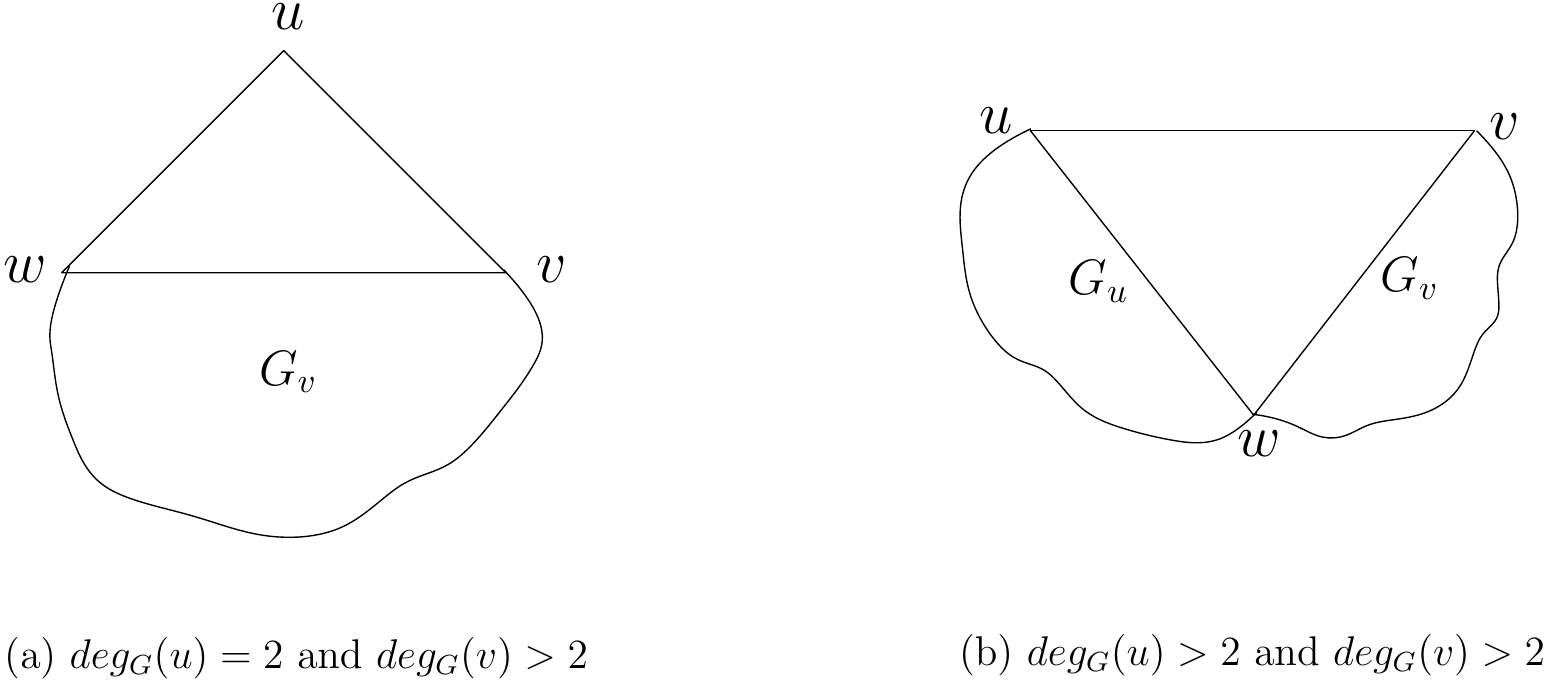}
 \caption{$G_u$ is the $uv$ segment of the graph that contains $u$ and $G_v$ is the $uv$ segment that contains $v$. In (a), $G_u$ is a single edge. }
 \label{fig:vc_computation_b}
\end{figure}
\begin{definition}[mvc and evc parameters]
 For a maximal outerplanar graph $G$ and an edge $uv$, the set of parameters  $\mathscr{M}(G,uv)=\{\mvc(G)$, $\mvc_{u}(G)$, $\mvc_{v}(G)$, $\mvc_{uv}(G)\}$
is called the (set of) $\textit{mvc parameters}$ of $G$ with respect to $uv$ and the set $\mathscr{E}(G,uv)=\{\evc(G)$, $\evc_{u}(G)$, $\evc_{v}(G)$, $\evc_{uv}(G)\}$ is called the (set of) $\textit{evc parameters}$ of $G$ with respect to $uv$.
\end{definition}

The following observation, which gives the mvc and evc parameters of a triangle, is the base case of the recursive computation of these parameters for larger graphs described in subsequent sections.

\begin{observation}\label{obs:basecase}
 Let $G$ be a triangle $uvw$. Consider the edge $uv$.  Then, (1) $\mvc(G)=\mvc_u(G)= \mvc_v(G)=\mvc_{uv}(G)=2$ (2) $\evc(G)=\evc_u(G)= \evc_v(G)=2$ and $\evc_{uv}(G)=3$.
 \end{observation}
 
\section{Computation of the mvc parameters}

Throughout this section, we assume that $G$ is a maximal outerplanar graph of at least four vertices,  $uv$ is an edge on the outer face of $G$ and  $w=\Delta(uv)$.  The results in this section show that the mvc parameters of $G$ with respect to the edge $uv$ - viz., $\mathscr{M}(G,uv)$,  can be computed in constant time if the mvc parameters of $G_u$ with respect to $uw$ - viz., $\mathscr{M}(G_u,uw)$ and the mvc parameters of $G_v$ with respect to $vw$ - viz., $\mathscr{M}(G_v,vw)$ are given. 

 The following observation is easy to see.
 
  \begin{observation}\label{obs:computemvcfix_deg2}
  If $deg_G(u)=2$ and $deg_G(v) >2$, then, 
  (1)~$\mvc(G)=\mvc_{vw}(G_v)$ (2)~$\mvc_u(G) = \mvc(G_v) +1$, (3)~$\mvc_v(G)=\mvc(G)$ and (4)~$\mvc_{uv}(G)=\mvc_v(G_v)+1$.
   \end{observation}   
   The following theorem is a consequence of Observation~\ref{obs:computemvcfix_deg2}.
\begin{theorem}
 \label{thm:vc_evc_deg2-1}
 Let $G$ be a maximal outerplanar graph and $uv$ be an edge on the outer face of $G$ such that $deg_G(u)=2$.
 Let $w=\Delta(uv)$. 
 \begin{enumerate}
  \item Given $\mathscr{M}(G_v, vw)$, it is possible to compute $\mvc(G)$ in constant time.
  \item Given $\mvc(G)$ and $\mathscr{M}(G_v, vw)$, it is possible to compute the remaining mvc parameters of $G$ with respect to $uv$ in constant time.
 \end{enumerate}
\end{theorem}

   Now, we look at the computation of $\mathscr{M}(G, uv)$ when $deg_G(u)>2$ and $deg_G(v)>2$. \\
   The following observation is easy to obtain. 
   
\begin{observation}\label{obs:computemvc1}
If $deg_G(u)>2$ and $deg_G(v) >2$, then\\ \indent $\mvc(G) \in \{\mvc(G_u)+\mvc(G_v)-1, \mvc(G_u)+\mvc(G_v)\}$.
\end{observation}
The next lemma gives a method to compute $\mvc(G)$ using $\mathscr{M}(G_u, uw)$ and $\mathscr{M}(G_v, vw)$. 
\begin{lemma}\label{lem:computemvc2}
 If $deg_G(u)>2$ and $deg_G(v) >2$, then \\ $\mvc(G)=\min\{\mvc_w(G_u)+\mvc_{vw}(G_v)-1, \mvc_{uw}(G_u)+\mvc_w(G_v)-1, \mvc(G_u)+\mvc(G_v)\}$. 
  \end{lemma}
  
  \begin{proof}
   It can be easily seen that 
   $\mvc(G)\le \min\{\mvc_w(G_u)+\mvc_{vw}(G_v)-1, \mvc_{uw}(G_u)+\mvc_w(G_v)-1, \mvc(G_u)+\mvc(G_v)\}$. By Observation~\ref{obs:computemvc1}, we have the following two cases to consider.
   
   If $\mvc(G)=\mvc(G_u)+\mvc(G_v)-1$, then both $G_u$ and $G_v$ have a minimum vertex cover with $w$. Further, either $\mvc_{uw}(G_u)=\mvc(G_u)$ or $\mvc_{vw}(G_v)=\mvc(G_v)$. Hence, it is clear that 
   $\mvc(G)=\mvc(G_u)+\mvc(G_v)-1 = \min\{\mvc_w (G_u ) + \mvc_{vw} (G_v )-1, \mvc_{uw} (G_u) + \mvc_w (G_v)-1\}$. From this, the result follows.
   
   If $\mvc(G) = \mvc(G_u ) + \mvc(G_v )$, then by the upper bound shown in the beginning of this proof, $\mvc(G_u ) + \mvc(G_v )=\mvc(G) \le \min\{\mvc_w (G_u ) + mvc_{vw}(G_v )-1, \mvc_{uw}(G_u ) + \mvc_w(G_v)-1\}$. From this, the result follows.\qed

  \end{proof}
  
  The next lemma shows that given $\mvc(G)$,  $\mathscr{M}(G_u, uw)$ and $\mathscr{M}(G_v, vw)$, the values of  $\mvc_u(G)$, $\mvc_v(G)$ and $\mvc_{uv}(G)$ are computable in constant time. 
  
 \begin{lemma}\label{lem:computemvcfix_gen-1}
Let $deg_G(u)>2$ and $deg_G(v)>2$. If $\mvc(G)=\mvc(G_u)+\mvc(G_v)-1$, then  \\
  $\mvc_u(G)=\mvc_{uw}(G_u)+\mvc(G_v)-1$, \\$\mvc_v(G)=\mvc(G_u)+\mvc_{vw}(G_v)-1$ and \\
  $\mvc_{uv}(G)=\mvc_{uw}(G_u)+\mvc_{vw}(G_v)-1$. 
\end{lemma}

\begin{proof}
Let $\mvc(G_u)=k_u$ and $\mvc(G_v)=k_v$.
Since $\mvc(G)=k_u+k_v-1$, both $G_u$ and $G_v$ have a minimum vertex cover with $w$. Let $k=\mvc_u(G)$. Since $\mvc_w(G_v)=k_v$, we can observe that $G$ has a vertex cover $C$ of size $k$ such that $u,w \in C$. Now, it can be easily seen that $C$ is the union of a vertex cover of $G_u$ containing 
 both $u$ and $w$ and a minimum vertex cover of $G_v$ with $w$. Therefore, 
 $\mvc_u(G)=\mvc_{uw}(G_u)+\mvc(G_v)-1$. Similarly, it is easy to see that $\mvc_v(G)=\mvc(G_u)+\mvc_{vw}(G_v)-1$ and 
 $\mvc_{uv}(G)=\mvc_{uw}(G_u)+\mvc_{vw}(G_v)-1$.
\qed
\end{proof}
 \begin{lemma}\label{lem:computemvcfix_gen-2}
Let $deg_G(u)>2$ and $deg_G(v)>2$. If $\mvc(G)=\mvc(G_u)+\mvc(G_v)$, then \\
      $\mvc_u(G)= \min\{\mvc_u(G_u)+\mvc(G_v), \mvc(G_u)+\mvc_w(G_v), \mvc(G)+1\}$,\\ 
    $\mvc_v(G)=\min\{ \mvc(G_u)+\mvc_v(G_v), \mvc_w(G_u)+\mvc(G_v), \mvc(G)+1\}$ and \\
    $\mvc_{uv}(G)=\min\{\mvc_u(G_u)+\mvc_{v}(G_v), \mvc(G)+1\}$.
\end{lemma}
\begin{proof}
Let $\mvc(G_u)=k_u$ and $\mvc(G_v)=k_v$.\\ 
First we prove that  $\mvc_u(G)= \min\{\mvc_u(G_u)+\mvc(G_v), \mvc(G_u)+\mvc_w(G_v), \mvc(G)+1\}$.  It can be easily seen that 
  $\mvc_u(G) \le \min\{\mvc_u(G_u)+\mvc(G_v), \mvc(G_u)+\mvc_w(G_v), \mvc(G)\\+1\}$. 
  Now, we need to show that 
  $\mvc_u(G) \ge \min\{\mvc_u(G_u)+\mvc(G_v), \mvc(G_u)+\mvc_w(G_v),\\ \mvc(G)+1\}$. 
  We know that $\mvc_u(G) \in \{\mvc(G), \mvc(G) +1\}$. If $\mvc_u(G) = \mvc(G) +1$, then the proof is immediate, because $\mvc_u(G)$ is equal to
one of the three quantities listed and so it is greater than or equal to their minimum. 
Otherwise, we have $\mvc_u(G) = \mvc(G)=\mvc(G_u)+\mvc(G_v)$. In this case, it is enough to show that among the three parameters, 
the first or the second will be equal to $\mvc(G_u)+\mvc(G_v)$. For contradiction, suppose this was not the case. Then, $\mvc_u(G_u)=k_u+1$ and $\mvc_w(G_v)=k_v+1$. 
Consider any minimum vertex cover $S$ of $G$ such that $u \in S$. If $w \in S$, then $|S| \ge \mvc_u(G_u)+\mvc_w(G_v)-1 \ge k_u + k_v+1$, a contradiction.
If  $w \notin S$, again $|S| \ge \mvc_u(G_u)+ \mvc(G_v)\ge k_u + k_v+1$, a contradiction. Hence, we are done.
The proof for $\mvc_v(G)$ is symmetric to the previous case.

Now, we show that $\mvc_{uv}(G)=\min\{\mvc_u(G_u)+\mvc_{v}(G_v), \mvc(G)+1\}$. It is easy to see that 
  $\mvc_{uv}(G) \le \min\{\mvc_u(G_u)+\mvc_{v}(G_v), \mvc(G)+1\}$. 
  We know that $\mvc_{uv}(G) \in \{\mvc(G), \mvc(G) +1\}$.
  Suppose $\mvc_{uv}(G)=\mvc(G)=k_u+k_v$. In this case, it suffices to show that $\mvc_u(G_u)= k_u$ and $\mvc_v(G_v) = k_v$. For contradiction, 
  assume $\mvc_u(G_u)=k_u+1$. Then, $\mvc_w(G_u)=k_u$. Since, $\mvc_{uv}(G)=k_u+k_v$, $\mvc_{vw}(G_v)=k_v$. Then, $\mvc(G)=k_u+k_v-1$, 
  which is a contradiction. Hence, $\mvc_u(G_u)\le k_u$. Similarly, $\mvc_v(G_v)\le k_v$ as required. 
  Now, suppose $\mvc_{uv}(G)=\mvc(G)+1=k_u+k_v+1$. In this case, it suffices to show that $\mvc_u(G_u)+\mvc_v(G_v)\ge \mvc_{uv}(G)=\mvc(G)+1$, 
  which is straightforward to obtain. \qed
\end{proof}

From Lemma~\ref{lem:computemvc2}, Lemma~\ref{lem:computemvcfix_gen-1} and Lemma~\ref{lem:computemvcfix_gen-2}, the following theorem is immediate.

 \begin{theorem}\label{thm:vc_evc_general-1}
 Let $G$ be a maximal outerplanar graph and $uv$ be an edge on the outer face of $G$ such that $deg_G(u)>2$ and $deg_G(v)>2$.
 Let $w=\Delta(uv)$.
 \begin{enumerate}
  \item Given $\mathscr{M}(G_u, uw)$ and $\mathscr{M}(G_v, vw)$, it is possible to compute $\mvc(G)$ in constant
  time.
  \item Given $\mvc(G)$, $\mathscr{M}(G_u, uw)$ and $\mathscr{M}(G_v, vw)$, it is possible to compute the 
  remaining mvc parameters of $G$ with respect to $uv$   in constant time.
  \end{enumerate}
  \end{theorem}
  
  \section{Bounds on the evc parameters}\label{bounds}
In this section, we will derive some bounds on the evc parameters of a maximal outerplanar graph. These bounds will be used in the next section for the recursive computation of the evc parameters. Throughout this section, we consider $G$ to be a maximal outerplanar graph on at least four vertices, $uv$ an edge on its outer face and $w=\Delta(uv)$. 

First, we derive some bounds for $\mathscr{E}(G, uv)$ when $deg_G(u)=2$.
 \begin{lemma}\label{lem:deg2}
If $u$ is a degree-2 vertex in $G$, then:
 \begin{enumerate}
  \item $\evc(G) \le \evc_u(G) = \evc(G_v)+1$
  \item $\evc_{uv}(G) = \evc_v(G_v)+1$ 
   \item $\evc(G) \ge max\{\mvc(G_v)+1, \evc_{vw}(G_v)\}$
  \item $\evc_v(G)=\text{max}\{\mvc_{uv}(G), \evc(G)\}$
 \end{enumerate}

\end{lemma}

\begin{proof}
Let $\evc(G)=k$.
 \begin{enumerate}
  \item First inequality is straightforward. So, it suffices to show that $\evc_u(G) = \evc(G_v)+1$. From 
  Proposition~\ref{prop:evc-subset}, $\max_{x\in V(G_v)} \mvc_x(G_v)=\evc(G_v)$. Hence, $\max_{x\in V(G_v)} \mvc_{xu}(G)= \evc(G_v)+1$. Hence, by
  Proposition~\ref{prop:evc-subset}, 
  we get $\evc_u(G)= \evc(G_v)+1$.
  \item Similar to the previous case, by Proposition~\ref{prop:evc-subset}, $\max_{x\in V(G_v)} \mvc_x(G_v)=\evc(G_v)$. Hence, $\max_{x\in V(G_v)} \mvc_{xuv}(G)\le \evc(G_v)+1$. Hence, by
  Proposition~\ref{prop:evc-subset}, 
  we get $\evc_{uv}(G)\le \evc(G_v)+1$. By Proposition~\ref{prop:occ_edge_vertices} and Part~1 of this lemma, $\evc_{uv}(G)\ge \evc_u(G)=\evc_v(G_v)+1$.
 Therefore, $\evc_{uv}(G) = \evc_v(G_v)+1$.
 \item By Observation~\ref{obs:computemvcfix_deg2}, we get $\mvc_u(G)=\mvc(G_v)+1$. By Proposition~\ref{prop:evc-subset}, since $\evc(G)\ge \mvc_u(G)$, 
  we have $\evc(G) \ge \mvc(G_v)+1$.\\
 By Proposition~\ref{prop:evc-subset}, for any vertex $x \in V(G)$, $\mvc_x(G)\le k$. Further, any vertex cover of $G$ contain at least two vertices among $u$, $v$ and $w$. Let $x$ be an arbitrary vertex in $G_v$ and $C$ be a vertex cover of $G$ of size $k$ containing $x$. If $C$ contains both $v$ and $w$, the vertex cover obtained by restricting $C$ to $V(G_v)$ gives a vertex cover of $G_v$ of size at most $k$. If $G$ does not contain $v$ (respectively, $w$), then $C'=C \setminus \{u\} \cup \{v\}$ (respectively, $C'=C \setminus \{u\} \cup \{w\}$) is a vertex cover of $G_v$ of size $k$ that contains $x$, $v$ and $w$. Therefore, for any vertex $x \in V(G_v)$, $\mvc_{xvw}(G_v)\le k$ containing $x$ along with both $v$ and $w$. By Proposition~\ref{prop:evc-subset}, $\evc_{vw}(G_v) \le k=\evc(G)$.
 \item By Proposition~\ref{prop:evc-subset}, $\evc_v (G) \ge \mvc_{uv}(G)$ and it is easy to see that $\evc_v (G) \ge \evc(G)$. Hence, $\evc_v(G)\ge \max\{\mvc_{uv}(G), \evc(G)\}$. 
 Now, we need to show that $\evc_v(G)\le \max\{\mvc_{uv}(G), \evc(G)\}$.
 By Proposition~\ref{prop:occ_edge_vertices}, $\evc_v(G) \le \evc(G)+1$. 
 If $\evc_v(G)=\evc(G)$, then we are done. Hence, suppose $\evc_v(G) = \evc(G)+1$. In this case, it is enough to show that $\mvc_{uv}(G) \ge \evc(G)+1$. By Proposition~\ref{prop:evc-subset}, we know that $\evc(G)=\max_{x \in V(G)}\mvc_{x}(G)$. Since degree of $u$ is $2$, this implies that for any vertex $x \in V(G)\setminus \{u\}$, $\mvc_{xv}(G)\le \evc(G)$. If $\mvc_{uv}(G)\le \evc(G)$, then by Proposition~\ref{prop:evc-subset}, $\evc_v(G)=\evc(G)$ which is a contradiction. Hence, $\mvc_{uv}(G)\ge \evc(G)+1$ as required.\qed
 \end{enumerate} 
\end{proof}

Now, we derive some upper bounds for $\mathscr{E}(G, uv)$ when $deg_G(u)>2$ and $deg_G(v)>2$.
\begin{proposition}\label{prop:mvc_x}
 \begin{enumerate}
  \item $\max_{x \in V(G_v)}\mvc_x(G) \le \min\{\mvc_w(G_u)+\evc_{vw}(G_v)-1,\mvc_{uw}(G_u)+\evc_w(G_v)\\-1, \mvc(G_u)+\evc(G_v)\}$
  \item $\max_{x \in V(G_u)}\mvc_x(G) \le \min\{\evc_{uw}(G_u)+\mvc_{w}(G_v)-1,\evc_{w}(G_u)+\mvc_{vw}(G_v)-1, \evc(G_u)+\mvc(G_v)\}$.
 \end{enumerate}
\end{proposition}
\begin{proof}
Since Parts~1 and 2 are symmetric, we only prove Part~1. Consider a vertex $x \in V(G_v)$. It can be easily seen that $\mvc_x(G) \le \min\{\mvc_w(G_u)+\mvc_{xvw}(G_v)-1,\mvc_{uw}(G_u)+\mvc_{xw}(G_v)-1\}$. By Proposition~\ref{prop:evc-subset}, we get $\mvc_x(G) \le \min\{\mvc_w(G_u)+\evc_ {vw}(G_v)-1,\mvc_{uw}(G_u)+\evc_w(G_v)-1\}$. Now, it remains to prove that $\mvc_x(G) \le \mvc(G_u)+\evc(G_v)$. If $\mvc_x(G_v)=\mvc_{xv}(G_v)$, then $\mvc_x(G) \le \mvc(G_u)+\mvc_{xv}(G_v)=\mvc(G_u)+\mvc_x(G_v)\le \mvc(G_u)+\evc(G_v)$ by Proposition~\ref{prop:evc-subset}. Otherwise, it is the case that $\mvc_x(G_v)=\mvc_{xw}(G_v)$. This implies $\mvc_x(G) \le \mvc_{uw}(G_u)+\mvc_{xw}(G_v)-1$. Since, $\mvc_{uw}(G)\le \mvc(G_u)+1$, we get $\mvc_x(G) \le \mvc(G_u)+\mvc_{xw}(G_v)=\mvc(G_u)+\mvc_x(G_v) \le \mvc(G_u)+\evc(G_v)$ by Proposition~\ref{prop:evc-subset}. \qed
\end{proof}

Lemma~\ref{lem:deg>2-11} gives an upper bound for $\evc(G)$ when $deg_G(u)>2$ and $deg_G(v)>2$.
\begin{lemma}\label{lem:deg>2-11}
 If $deg_{G}(u)>2$ and $deg_{G}(v) > 2$, then:
 \begin{enumerate}
  \item $\evc(G) \le \min \{\mvc(G)+1, \evc_w(G_u)+\evc_w(G_v)-1\}$
  \item $\evc(G) \le \max\{U, V\}$ where\\ $U=\min\{\evc_{uw}(G_u)+\mvc_{w}(G_v)-1,\evc_{w}(G_u)+\mvc_{vw}(G_v)-1, \evc(G_u)+\mvc(G_v)\}$ and \\$V=\min\{\mvc_w(G_u)+\evc_{vw}(G_v)-1,\mvc_{uw}(G_u)+\evc_w(G_v)-1, \mvc(G_u)+\evc(G_v)\}$
 \end{enumerate}
\end{lemma}
\begin{proof}
\begin{enumerate}
 \item  By Proposition~\ref{prop:evc-subset}, it is easy to see that $\evc(G) \le \min\{\mvc(G)+1, \evc_w(G_u)+\evc_w(G_v)-1\}$. 
 \item By Proposition~\ref{prop:mvc_x}, we have $\max_{x \in V(G)} \mvc_x(G) \le \max(U, V)$. Hence, by Proposition~\ref{prop:evc-subset}, we get $\evc(G) \le \max(U, V)$.
\end{enumerate}
\end{proof}
\begin{proposition}\label{prop:mvc_vx}
 \begin{enumerate}
  \item $\max_{x \in V(G_v)}\mvc_{xv}(G) \le \min\{\mvc(G_u)+\evc_{v}(G_v),\mvc_{w}(G_u)+\evc_{vw}(G_v)-1\}$
  \item $\max_{x \in V(G_u)}\mvc_{xv}(G) \le \min\{\evc(G_u)+\mvc_{v}(G_v),\evc_{w}(G_u)+\mvc_{vw}(G_v)-1\}$
 \end{enumerate}
\end{proposition}
\begin{proof}
 \begin{enumerate}
  \item It is easy to see that for any vertex 
 $x \in V(G_v)$, $\mvc_{xv}(G) \le \mvc_w(G_u)+\mvc_{xvw}(G_v)-1$ and $\mvc_{xv}(G) \le  \mvc(G_u)+\mvc_{xv}(G_v)$. 
  By Proposition~\ref{prop:evc-subset}, we get the result.
  \item  It is easy to see that for any vertex $x \in V(G_u)$, $\mvc_{xv}(G) \le \mvc_x(G_u)+\mvc_v(G_v)$ and $\mvc_{xv}(G) \le \mvc_{xw}(G_u)+\mvc_{vw}(G_v)-1$. By Proposition~\ref{prop:evc-subset}, we get the result. \qed
 \end{enumerate}

\end{proof}  
Lemma~\ref{lem:deg>2-12} gives an upper bound for $\evc_v(G)$ when $deg_G(u)>2$ and $deg_G(v)>2$.
\begin{lemma}\label{lem:deg>2-12}
 If $deg_{G}(u)>2$ and $deg_{G}(v) > 2$, then:
 \begin{enumerate}
  \item $\evc_v(G) \le \min\{\evc(G)+1, \mvc_v(G)+1,\evc_w(G_u)+\evc_{vw}(G_v)-1\}$
  \item $\evc_v(G) \le \max\{U,V\}$, where $U=  \min \{ \evc(G_u)+\mvc_v(G_v), \evc_w(G_u)+ \mvc_{vw}(G_v)-1\}$ and \\$V=\min\{\mvc(G_u)+\evc_{v}(G_v),\mvc_{w}(G_u)+\evc_{vw}(G_v)-1\}$
 \end{enumerate}
\end{lemma}
\begin{proof}
\begin{enumerate}
 \item By Proposition~\ref{prop:occ_edge_vertices}, we know that $\evc_v(G) \le \evc(G)+1$. By Proposition~\ref{prop:mvc_vx}, $\max_{x\in V(G_u)}\\ \mvc_{xv}(G) \le \evc_{w}(G_u)+\mvc_{vw}(G_v)-1 \le \evc_w(G_u)+\evc_{vw}(G_v)-1$ and $\max_{x\in V(G_v)}\mvc_{xv}(G) \le \mvc_{w}(G_u)+\evc_{vw}(G_v)-1 \le \evc_w(G_u)+\evc_{vw}(G_v)-1$. By Proposition~\ref{prop:evc-subset}, 
  we get $\evc_v(G)\le  \evc_w(G_u)+\evc_{vw}(G_v)-1$. By Proposition~\ref{prop:evc-subset}, 
  we also get $\evc_v(G)\le \mvc_v(G)+1$. Hence, we are done.
  \item  By Proposition~\ref{prop:mvc_vx}, we have $\max_{x \in V(G_u)}\mvc_{vx}(G) \le U$ and $\max_{y \in V(G_v)}\mvc_{vy}(G) \le V$. Therefore, by Proposition~\ref{prop:evc-subset}, 
  we get $\evc_v(G)= \max_{x \in V(G)}\mvc_{vx}(G) \le  \max \{U, V\}$. 
\end{enumerate}
\end{proof}

Proposition~\ref{prop:mvc_ux} and  Lemma~\ref{lem:deg>2-122} are symmetric to Proposition~\ref{prop:mvc_vx} and  Lemma~\ref{lem:deg>2-12} respectively.
\begin{proposition}\label{prop:mvc_ux}
 \begin{enumerate}
  \item $\max_{x \in V(G_u)}\mvc_{xu}(G)\le \min\{\evc_u(G_u)+\mvc(G_v),\evc_{uw}(G_u)+\mvc_{w}(G_v)-1\}$
  \item $\max_{x \in V(G_v)}\mvc_{xu}(G) \le \min\{\mvc_u(G_u)+\evc(G_v),\mvc_{uw}(G_u)+\evc_{w}(G_v)-1\}$
 \end{enumerate}
\end{proposition}
\begin{lemma}\label{lem:deg>2-122}
 If $deg_{G}(u)>2$ and $deg_{G}(v) > 2$, then:
 \begin{enumerate}
  \item $\evc_u(G) \le \min\{\evc(G)+1, \mvc_u(G)+1,\evc_{uw}(G_u)+\evc_{w}(G_v)-1\}$
  \item $\evc_u(G) \le \max\{U,V\}$, where $U=\min\{\evc_u(G_u)+\mvc(G_v),\evc_{uw}(G_u)+\mvc_{w}(G_v)-1\} $ and \\$V= \min \{ \mvc_u(G_u)+\evc(G_v), \mvc_{uw}(G_u)+ \evc_{w}(G_v)-1\}$
 \end{enumerate}
\end{lemma}
\begin{proposition}\label{prop:mvc_uvx}
 \begin{enumerate}
    \item $\max_{x \in V(G_u)}\mvc_{xuv}(G) \le \min\{\evc_{uw}(G_u)+\mvc_{vw}(G_v)-1, \evc_u(G_u)+\mvc_v(G_v)\}$
  \item $\max_{x \in V(G_v)}\mvc_{xuv}(G) \le \min\{\mvc_{uw}(G_u)+\evc_{vw}(G_v)-1, \mvc_u(G_u)+\evc_v(G_v)\}$
 \end{enumerate}
\end{proposition}
\begin{proof}
\begin{enumerate}
 \item  Consider any vertex $x \in V(G_u)$. 
By Proposition~\ref{prop:evc-subset}, we have $\mvc_{xuw}(G_u) \le \evc_{uw}(G_u)$. This implies $\mvc_{xuv}(G) \le \evc_{uw}(G_u)+\mvc_{vw}(G_v)-1$.
Similarly, by Proposition~\ref{prop:evc-subset}, we have $\mvc_{xu}(G_u) \le \evc_{u}(G_u)$. This implies $\mvc_{xuv}(G) \le \evc_{u}(G_u)+\mvc_v(G_v)$.
Thus, for any $x \in V(G_u)$, we have $\mvc_{xuv}(G) \le \min\{\evc_{uw}(G_u)+\mvc_{vw}(G_v)-1, \evc_u(G_u)+\mvc_v(G_v)\}$.
\item This case is symmetric to the previous one.
\end{enumerate}

\end{proof}

Lemma~\ref{lem:deg>2-2} gives an upper bound for $\evc_{uv}(G)$ when $deg_G(u)>2$ and $deg_G(v)>2$.
\begin{lemma}\label{lem:deg>2-2}
 If $deg_{G}(u)>2$ and $deg_{G}(v) > 2$, then:
 \begin{enumerate}
  \item  $\evc_{uv}(G)\le \min \{\evc(G)+1, \mvc_{uv}(G)+1, \max\{P_1, P_2\}\}$, 
  where \\$P_1=\min\{\evc_{uw}(G_u)+\mvc_{vw}(G_v)-1, \evc_u(G_u)+\mvc_v(G_v)\}$ and \\
  $P_2=\min\{\mvc_{uw}(G_u)+\evc_{vw}(G_v)-1, \mvc_u(G_u)+\evc_v(G_v)\}$ 
\item     $\evc_{uv}(G)\le \min\{Q, \evc_{uw}(G_u)+\evc_{vw}(G_v)-1,
\evc_u(G_u)+\evc(G_v), \evc_v(G_v)+\evc(G_u)\}$, where $Q=\max\{\evc_{u}(G_u)+\evc_{vw}(G_v)-1, \evc_{u}(G_u)+\mvc_v(G_v)\}$
 \end{enumerate}
\end{lemma}
\begin{proof}
 \begin{enumerate}
  \item By Proposition~\ref{prop:occ_edge_vertices}, we know that $\evc_{uv}(G) \le \evc(G)+1$. By Proposition~\ref{prop:evc-subset}, we know that $\evc_{uv}(G) \le \mvc(G)+1$.  By Proposition~\ref{prop:evc-subset} and Proposition~\ref{prop:mvc_uvx}, we get $\evc_{uv}(G) =\max_{x \in V(G)}\\ \mvc_{xuv}(G) \le max(P_1, P_2)$.
\item It can be easily seen that $\evc_{uv}(G) \le \evc_{uw}(G_u)+\evc_{vw}(G_v)-1$. Now, we show that $\evc_{uv}(G) \le Q$. By Proposition~\ref{prop:mvc_uvx}, we get $\max_{x \in V(G)}\mvc_{xuv}(G) \le \max\{\mvc_{uw}(G_u)+\evc_{vw}(G_v)-1, \evc_{u}(G_u)+\mvc_v(G_v)\}$.  Hence, by Proposition~\ref{prop:evc-subset}, we get $\evc_{uv}(G) \le Q$. \\Next, we show that $\evc_{uv}(G) \le \evc_u(G_u)+\evc(G_v)$.
By Proposition~\ref{prop:evc-subset} and Proposition~\ref{prop:mvc_uvx}, we get $\max_{x \in V(G_u)} \mvc_{xuv}(G) \le \evc_u(G_u)+\mvc_v(G_v) \le \evc_u(G_u)+\evc(G_v)$.
Since $\mvc_{uw}(G_u)\le \evc_u(G_u)$ by Proposition~\ref{prop:evc-subset} and $\evc_{vw}(G_v)\le \evc(G_v)+1$ by Proposition~\ref{prop:occ_edge_vertices}, using Proposition~\ref{prop:mvc_uvx},  we get  $\max_{x \in V(G_v)}\mvc_{xuv}(G) \le \mvc_{uw}(G_u)+\evc_{vw}(G_v)-1 \le \evc_u(G_u)+\evc(G_v)$. Hence, by Proposition~\ref{prop:evc-subset}, we get $\evc_{uv}(G) \le \evc_u(G_u)+\evc(G_v)$. Symmetrically, we can prove that $\evc_{uv}(G) \le \evc(G_u)+\evc_{v}(G_v)$.
 \end{enumerate}
\end{proof}
The following lemma gives some lower bounds for $\mathscr{E}(G, uv)$ when $deg_G(u)>2$ and $deg_G(v)>2$.

\begin{lemma}\label{lem:deg>2-3}
If $deg_{G}(u)>2$ and $deg_{G}(v) > 2$, then:
 \begin{enumerate}
 \item $\evc(G) \ge \max\{\evc_w(G_u)+\mvc(G_v)-1, \mvc(G_u)+\evc_w(G_v)-1\}$
   \item $\evc_u(G)\ge \max\{\evc(G), \mvc_{uw}(G_u)+\evc(G_v)-1, \evc_{uw}(G_u)+\mvc(G_v)-1,\\ \evc_u(G_u)+\mvc_w(G_v)-1\}$ 
  \item $\evc_{uv}(G) \ge \max\{\evc(G), \evc_u(G), \evc_v(G), \mvc_{uw}(G_u)+\evc_v(G_v)-1,\\ \evc_u(G_u)+\mvc_{vw}(G_v)-1, \evc_{uw}(G_u)+\mvc_v(G_v)-1, \mvc_u(G_u)+\evc_{vw}(G_v)-1\}$  
 \end{enumerate}
\end{lemma}
\begin{proof}
 \begin{enumerate}
 \item For any vertex $x \in V(G_u)$, consider a minimum vertex cover $S_x$ of $G$ of size $\mvc_x(G)$ containing $x$. It is easy to see that $|S_x \cap V(G_v)| \ge \mvc(G_v)$. If $w \in S_x$, then $|S_x \cap V(G_u)|\ge \mvc_{wx}(G_u)$ and hence, $\mvc_x(G) \ge \mvc_{wx}(G_u)+\mvc(G_v)-1$. If $w \notin S_x$, then $|S_x \cap V(G_u)| \ge \mvc_{x}(G_u)$. Therefore, $\mvc_x(G) \ge \mvc_{x}(G_u)+\mvc(G_v)$. Since $\mvc_{wx}(G_u) \le \mvc_x(G_u)+1$, in this case also we get $\mvc_x(G) \ge \mvc_{wx}(G_u)+\mvc(G_v)-1$. By Proposition~\ref{prop:evc-subset}, we get 
 $\evc(G)\ge \max_{x \in V(G_u)}\mvc_x(G) \\ \ge \max_{x \in V(G_u)}\mvc_{wx}(G_u)+\mvc(G_v)-1 \ge \evc_w(G_u)+\mvc(G_v)-1$.
 Symmetrically, we get $\evc(G) \ge \mvc(G_u)+\evc_w(G_v)-1$. From this, the result follows.
   \item By Proposition~\ref{prop:evc-subset}, it can be easily seen that $\evc_u(G)\ge \evc(G)$.
   By Proposition~\ref{prop:evc-subset}, there exists a vertex $x \in V(G_v)$ such that $\mvc_x(G_v)=\evc(G_v)$. 
   Let $S_{xu}$ be a vertex cover of $G$ of size $\mvc_{xu}(G)$ containing $x$ and $u$. If $w \in S_{xu}$, then it can be seen that $|S_{xu}| \ge \mvc_{uw}(G_u)+\mvc_x(G_v)-1=\mvc_{uw}(G_u)+\evc(G_v)-1$. Otherwise, we get $|S_{xu}| \ge \mvc_{u}(G_u)+\evc(G_v)\ge \mvc_{uw}(G_u)-1+\evc(G_v)$. Hence, by Proposition~\ref{prop:evc-subset}, we get $\evc_u(G) \ge \mvc_{uw}(G_u)+\evc(G_v)-1$. 
     Next, we to show that $\evc_{u}(G) \ge \evc_{uw}(G_u)+\mvc(G_v)-1$. For every vertex $x \in V(G_u)$, $\mvc_{xu}(G) \ge \mvc_{xuw}(G_u)+\mvc(G_v)-1$. 
     By Proposition~\ref{prop:evc-subset}, it can be seen that 
     $\evc_{u}(G) \ge \max_{x \in V(G_u)} \mvc_{xu}(G) \ge \max_{x \in V(G_u)}\mvc_{xuw}(G_u)+\mvc(G_v)-1=
     \evc_{uw}(G_u)+\mvc(G_v)-1$.      
   Now, it remains to show that $\evc_u(G) \ge \evc_u(G_u)+\mvc_w(G_v)-1$.
   By Proposition~\ref{prop:evc-subset}, there exists a vertex $x \in V(G_u)$ such that $\mvc_{xu}(G_u)=\evc(G_u)$. 
   Let $S_{xu}$ be a vertex cover of $G$ of size $\mvc_{xu}(G)$ containing $x$ and $u$. If $w \in S_{xu}$, then it can be seen that $|S_{xu}| \ge \mvc_{xuw}(G_u)+\mvc_w(G_v)-1=\evc_{u}(G_u)+\mvc_w(G_v)-1$. Otherwise, we get $|S_{xu}| \ge \evc_{u}(G_u)+\mvc(G_v)\ge \evc_{u}(G_u)+\mvc_w(G_v)-1$. Hence, by Proposition~\ref{prop:evc-subset}, we get $\evc_u(G) \ge \evc_{u}(G_u)+\mvc_w(G_v)-1$.
\item  By Proposition~\ref{prop:evc-subset}, it can be easily seen that $\evc_{uv}(G) \ge \max\{\evc(G), \evc_u(G), \evc_v(G)\}$. Similar to the proof of $\evc_u(G) \ge \mvc_{uw}(G_u)+\evc(G_v)$ in the previous case, we get $\evc_{uv}(G) \ge \mvc_{uw}(G_u)+\evc_v(G_v)-1$ and $\evc_{uv}(G) \ge \evc_u(G_u)+\mvc_{vw}(G_v)-1$.
Similar to the proof of $\evc_u(G) \ge \evc_{uw}(G_u)+\mvc(G_v)-1$ in the previous case, we get  $\evc_{uv}(G) \ge \evc_{uw}(G_u
)+\mvc_v(G_v)-1$ and $ \evc_{uv}(G) \ge \mvc_u(G_u)+\evc_{vw}(G_v)-1\}$.\qed
 \end{enumerate}
\end{proof}
  \section{Computation of the evc parameters}
  
Let $G$ be a maximal outerplanar graph with at least four vertices, $uv$ be an edge on its outer face and $w=\Delta(uv)$. In this section, we describe the method of computing $\mathscr{E}(G, uv)$ using the bounds obtained in Section~\ref{bounds}.
 \subsection{Computing $\mathscr{E}(G, uv)$ when $deg_G(u)=2$}
In this subsection, we consider the case when $u$ is a degree-$2$ vertex in $G$.
 \begin{lemma} \label{lem:necessary}
 Let $\mvc(G_v) \ne \evc(G_v)$. Then, 
 $\evc(G)=\evc(G_v)$ if and only if $\evc_{vw}(G_v)=\evc(G_v)$.
\end{lemma} 
\begin{proof}
Let $\mvc(G_v) < \evc(G_v)\le \evc_{vw}(G_v)$. 
Then, by Proposition~\ref{prop:evc-subset}, $\evc(G_v)=\mvc(G_v)+1$. Lemma~\ref{lem:deg2} Part 3, $\evc(G) \ge \evc_{vw}(G_v)$. 

To prove the forward direction, assume $\evc(G)=\evc(G_v)$. 
This gives $\evc(G_v)=\evc(G) \ge \evc_{vw}(G_v)$. This implies $\evc(G_v)=\evc_{vw}(G_v)$. Now, to prove the reverse direction, assume $\evc_{vw}(G_v)=\evc(G_v)$. It is enough to show that $\evc(G) \le\evc_{vw}(G_v)$. 
By Proposition~\ref{prop:evc-subset}, for any vertex $x \in V(G_v)$, we get $\mvc_{xvw}(G_v) \le \evc(G_v)$ which implies $\mvc_x(G) \le \evc(G_v)$.  Moreover, since $\mvc_u(G) \le \mvc(G_v)+1= \evc(G_v)$, by Proposition~\ref{prop:evc-subset}, $\evc(G)\le \evc(G_v)=\evc_{vw}(G_v)$.\qed
\end{proof}
  
\begin{lemma}\label{lem:evc_deg2}
 $\evc(G)=\text{max}\{\mvc(G_v)+1, \evc_{vw}(G_v)\}$.
 \end{lemma}
 \begin{proof}
   We prove this lemma by splitting into following two cases.
  \begin{enumerate}
   \item Suppose $\mvc(G_v)=\evc(G_v)$. In this case, by Lemma~\ref{lem:deg2}, Part 1 and Part 3, we have $\evc(G)= \evc(G_v)+1= \mvc(G_v)+1$.
  Further by Part 3 of Lemma~\ref{lem:deg2}, 
   $\evc_{vw}(G_v) \le \evc(G)$.
   Thus, we have $\evc(G)=\text{max}\{\mvc(G_v)+1, \evc_{vw}(G_v)\}$.
   \item Suppose $\mvc(G_v) < \evc(G_v)$. If $\evc_{vw}(G_v)=\evc(G_v)$, then by Lemma~\ref{lem:necessary}, 
   we have $\evc(G)=\evc(G_v)=\evc_{vw}(G_v)=\mvc(G_v)+1$.
   Else if $\evc_{vw}(G_v) \ne \evc(G_v)$, then by Lemma~\ref{lem:necessary} and  Proposition~\ref{prop:occ_edge_vertices}, we get $\evc(G)=\evc(G_v)+1=\evc_{vw}(G_v)$. In both the cases, 
   we can observe that $\evc(G)=\max\{\mvc(G_v)+1, \evc_{vw}(G_v)\}$.\qed
  \end{enumerate}

 \end{proof}
From Lemma~\ref{lem:deg2} and Lemma~\ref{lem:evc_deg2}, the following theorem is immediate.
\begin{theorem}
 \label{thm:vc_evc_deg2-2}
 Let $G$ be a maximal outerplanar graph and $u$ be a degree-2 vertex in $G$ with neighbors $v$ and $w$.
 \begin{enumerate}
   \item Given $\mvc(G_v)$ and $\mathscr{E}(G_v, vw)$, it is possible to compute $\evc(G)$ in constant time.
 \item Given $\evc(G)$, $\mathscr{M}(G, uv)$ and $\mathscr{E}(G_v, vw)$, it is possible to compute the remaining evc parameters of $G$ with respect to $uv$ in constant time.
 \end{enumerate}
  \end{theorem}  
  \subsection{Computing $\mathscr{E}(G, uv)$ when $deg_G(u)>2$ and $deg_G(v)>2$}
Now, we will see how $\mathscr{E}(G, uv)$ can be computed when $deg_G(u)>2$ and $deg_G(v)>2$ using $\mvc(G)$, $\mathscr{M}(G_u, uw)$, $\mathscr{M}(G_v, vw)$, $\mathscr{E}(G_u, uw)$ and $\mathscr{E}(G_v, vw)$. For this subsection, we will assume that $deg_G(u)>2$ and $deg_G(v)>2$. We will also assume that $\mvc(G_u)=k_u$ and $\mvc(G_v)=k_v$.

Lemmas~\ref{lem:algoevc-1}-\ref{lem:algoevc-3} handle the computation of $\evc(G)$.
\begin{lemma}\label{lem:algoevc-1}
If $\evc(G_u)=k_u$ and $\evc(G_v)=k_v$, then\\
 \indent $\evc(G) = \max\{\evc_w(G_u)+\mvc(G_v)-1, \mvc(G_u)+\evc_w(G_v)-1\}$.
\end{lemma}
\begin{proof}
First, suppose $\evc_w(G_u)=k_u$ and $\evc_w(G_v)=k_v$. By Lemma~\ref{lem:deg>2-11} Part~1, we have  $\evc(G) \le k_u+k_v-1$. By Observation~\ref{obs:computemvc1} and  
  Proposition~\ref{prop:evc-subset}, we get $\evc(G)=k_u+k_v-1$. Note that in this case $\max\{\evc_w(G_u)+\mvc(G_v)-1, \mvc(G_u)+\evc_w(G_v)-1\}=k_u+k_v-1$ as well. 
  
  Now, by Proposition~\ref{prop:occ_edge_vertices}, it remains to consider the case when $\evc_w(G_u)=k_u+1$ or $\evc_w(G_v)=k_v+1$. In this case, $\max\{\evc_w(G_u)+\mvc(G_v)-1, \mvc(G_u)+\evc_w(G_v)-1\}=k_u+k_v$. We will show that $\evc(G)=k_u+k_v$. From Lemma~\ref{lem:deg>2-11} Part~2, we get $\evc(G) \le \max\{\evc(G_u)+\mvc(G_v), \mvc(G_u)+\evc(G_v)\}= k_u+k_v$. By Lemma~\ref{lem:deg>2-3} Part~1, we get $\evc(G) \ge \max\{\evc_w(G_u)+\mvc(G_v)-1, \mvc(G_u)+\evc_w(G_v)-1\}$.  Hence, we are done.\qed
\end{proof}

\begin{lemma}\label{lem:algoevc-2}
If $\evc(G_u)=k_u$ and $\evc(G_v)=k_v+1$, then \\
 \indent 
     $\evc(G)=\min\{ \mvc(G)+1, \evc(G_u)+\evc_{vw}(G_v)-1$, $\mvc_{uw}(G_u)+\evc_w(G_v)-1\}$.
\end{lemma}
\begin{proof}
Since $\evc(G_v)=k_v+1$, by Lemma~\ref{lem:deg>2-3} Part~1, we get $\evc(G) \ge k_u+k_v$.
By Observation~\ref{obs:computemvc1} and Proposition~\ref{prop:evc-subset}, we get $\evc(G) \in \{k_u+k_v, k_u+k_v+1\}$. By Proposition~\ref{prop:occ_edge_vertices}, it is easy to see that $\min\{ \mvc(G)+1, \evc(G_u)+\evc_{vw}(G_v)-1$, $\mvc_{uw}(G_u)+\evc_w(G_v)-1\} \in \{k_u+k_v, k_u+k_v+1\}$. 

Let $U$ and $V$ be as defined in Lemma~\ref{lem:deg>2-11} Part~2. From our assumptions, it follows that $U \le k_u+k_v$ and $V \ge k_u+k_v$. Hence by  Lemma~\ref{lem:deg>2-11} Part~2, we get $\evc(G) \le V \le min\{ \evc(G_u)+\evc_{vw}(G_v)-1$, $\mvc_{uw}(G_u)+\evc_w(G_v)-1\}$. From this, using Proposition~\ref{prop:evc-subset}, we get $\evc(G) \le min\{\mvc(G)+1, \evc(G_u)+\evc_{vw}(G_v)-1$, $\mvc_{uw}(G_u)+\evc_w(G_v)-1\}$.

Hence, it suffices to show that when $\evc(G)=k_u+k_v$, then $min\{ \mvc(G)+1, \evc(G_u)+\evc_{vw}(G_v)-1$, $\mvc_{uw}(G_u)+\evc_w(G_v)-1\}=k_u+k_v$. For this, suppose $\evc(G)=k_u+k_v$. By Lemma~\ref{lem:deg>2-3} Part~1, we have $\evc_w(G_u) \le k_u+1$ and $\evc_w(G_v)\le k_v+1$. If $\mvc_{uw}(G_u)=k_u$, we get $\min\{\evc(G_u)+\evc_{vw}(G_v)-1$, $\mvc_{uw}(G_u)+\evc_w(G_v)-1\} \le k_u+k_v=\evc(G)$ as required. Therefore, let us assume that $\mvc_{uw}(G_u)=k_u+1$. In this case, it is enough to show that $\evc_{vw}(G_v)=k_v+1$. By Proposition~\ref{prop:evc-subset}, there exists a vertex $x \in V(G_v)$ such that $\mvc_x(G_v)=k_v+1$. Consider such a vertex $x$. By Proposition~\ref{prop:evc-subset}, we know there exits a vertex cover $S_x$ of $G$ such that $x \in S_x$ and $|S_x|=\evc(G)=k_u+k_v$. Since $\mvc_{uw}(G_u)=k_u+1$ and $\mvc_x(G_v)=k_v+1$, it is easy to see that $u \notin S_x$ and hence $v, w \in S_x$. Further, $|S_x \cap V(G_v)|=k_v+1$ and hence $\mvc_{xvw}(G_v)=k_v+1$. By Proposition~\ref{prop:evc-subset}, this implies that $\evc_{vw}(G_v)=k_v+1$, as required.\qed
\end{proof}

\begin{lemma}\label{lem:algoevc-3}
If $\evc(G_u)=k_u+1$ and $\evc(G_v)=k_v+1$, then \\
 \indent 
$\evc(G)=\min\{\mvc(G)+1,\max \{\evc_{uw}(G_u)+\mvc_w(G_v)-1, \mvc_w(G_u)+\evc_{vw}(G_v)-1\}\}$.
\end{lemma}
\begin{proof}
Since $\evc(G_v)=k_v+1$, by Lemma~\ref{lem:deg>2-3} Part~1, we get $\evc(G) \ge k_u+k_v$.
By Observation~\ref{obs:computemvc1} and Proposition~\ref{prop:evc-subset}, we get $\evc(G) \in \{k_u+k_v, k_u+k_v+1\}$. By Proposition~\ref{prop:occ_edge_vertices}, it is easy to see that $\min\{\mvc(G)+1,\max \{\evc_{uw}(G_u)+\mvc_w(G_v)-1, \mvc_w(G_u)+\evc_{vw}(G_v)-1\}\} \in \{k_u+k_v, k_u+k_v+1\}$. 

By Lemma~\ref{lem:deg>2-11}, we get $\evc(G)\le \min\{\mvc(G)+1, \max \{\evc_{uw}(G_u)+\mvc_w(G_v)-1, \evc_{vw}(G_v)+\mvc_w(G_u)-1\}\}$. Hence, it is enough to show that if $\evc(G)=k_u+k_v$, then $\min\{\mvc(G)+1, \max \{\evc_{uw}(G_u)+\mvc_w(G_v)-1, \evc_{vw}(G_v)+\mvc_w(G_u)-1\}\}=k_u+k_v$.

For this, suppose $\evc(G)=k_u+k_v$. Since $\evc(G)=k_u+k_v$, $\evc(G_u)=k_u+1$ and $\evc(G_v)=k_v+1$, by Lemma~\ref{lem:deg>2-3} Part~1, we get $\evc_w(G_u)=k_u+1$ and $\evc_w(G_v)=k_v+1$. Now, it is enough to show that $\mvc_w(G_u)=k_u$ and $\mvc_w(G_v)=k_v$. Since $\evc_w(G_u)=k_u+1$, by Proposition~\ref{prop:evc-subset}, there exists a vertex $x \in V(G_u)$ such that $\mvc_{xw}(G_u)= k_u+1$. By Proposition~\ref{prop:evc-subset}, we know $k_u+k_v=\evc(G) \ge \mvc_x(G)$. Therefore, it has to be the case that $\mvc_w(G_v)=k_v$. Symmetrically, we get $\mvc_w(G_u)=k_u$.\qed
\end{proof}

Lemmas~\ref{lem:case1algo-1}-~\ref{lem:case4algo} deal  with the computation of the remaining parameters in $\mathscr{E}(G,uv)$ using $\evc(G)$, $\mathscr{M}(G,uv)$, $\mathscr{M}(G_u,uw)$, $\mathscr{M}(G_v,vw)$, $\mathscr{E}(G_u,uw)$ and $\mathscr{E}(G_v,vw)$.
\begin{lemma} \label{lem:case1algo-1} 
If $\evc(G_u)=k_u$, $\evc(G_v)=k_v$ and $\evc(G)=k_u+k_v-1$, then\\ 
  \indent $\evc_{u}(G)=\evc_{uw}(G_u)+\mvc(G_v)-1$, 
   $\evc_{v}(G)=\evc_{vw}(G_v)+\mvc(G_u)-1$ and\\ 
  \indent   $\evc_{uv}(G)=min\{\mvc_{uv}(G)+1, \evc_{uw}(G_u)+\evc_{vw}(G_v)-1\}$.
\end{lemma}
\begin{proof}
     By Lemma~\ref{lem:deg>2-3} Part~1, we know $\evc(G) \ge \max\{\evc_w(G_u)+\mvc(G_v)-1, \mvc(G_u)+\evc_w(G_v)-1\}$. Therefore, we get $\evc_w(G_u)=k_u$ and $\evc_w(G_v)=k_v$. By Proposition~\ref{prop:evc-subset}, we know  $\mvc_{uw}(G_u)=k_u$ and $\mvc_{vw}(G_v)=k_v$. Hence, we get $\mvc_{uv}(G)=k_u+k_v-1$.   
     
     First, we will prove the expression for $\evc_u(G)$.  By Lemma~\ref{lem:deg>2-122}, $\evc_{u}(G) \le \evc_{uw}(G_u)+\evc_w(G_v)-1=\evc_{uw}(G_u)+\mvc(G_v)-1$. 
     By Lemma~\ref{lem:deg>2-3} Part~2, we get $\evc_{u}(G) \ge \evc_{uw}(G_u)+\mvc(G_v)-1$.
     Hence, $\evc_{u}(G) = \evc_{uw}(G_u)+\mvc(G_v)-1$.
 Similarly, we get $\evc_{v}(G)=\evc_{vw}(G_v)+\mvc(G_u)-1$. \\              
     Now, we show that $\evc_{uv}(G)= min\{\mvc_{uv}(G)+1, \evc_{uw}(G_u)+\evc_{vw}(G_v)-1\}$. By Lemma~\ref{lem:deg>2-2}, we have $\evc_{uv}(G)\le  min\{\mvc_{uv}(G)+1, \evc_{uw}(G_u)+\evc_{vw}(G_v)-1\}$.
     By Proposition~\ref{prop:evc-subset}, we know that $\evc_{uv}(G) \in \{\mvc_{uv}(G),\mvc_{uv}(G)+1\}$.
 Now, it is enough to show that if $\evc_{uv}(G) = \mvc_{uv}(G)$, then 
 $\evc_{uw}(G_u)+\evc_{vw}(G_v)-1 = \mvc_{uv}(G)$ Recall that $\mvc_{uv}(G)=k_u+k_v-1$. 
 Suppose $\evc_{uv}(G)=\mvc_{uv}(G)=k_u+k_v-1$. Then, by Lemma~\ref{lem:deg>2-3} Part~3, we get 
      $\evc_{uw}(G_u)=k_u$ and $\evc_{vw}(G_v)=k_v$. Therefore,
      $\evc_{uw}(G_u)+\evc_{vw}(G_v)-1=k_u+k_v-1$.\qed
      
\end{proof}

\begin{lemma} \label{lem:case1algo-2} 
If $\evc(G_u)=k_u$, $\evc(G_v)=k_v$ and $\evc(G)=k_u+k_v$, then\\ \indent  $\evc_u(G)=\evc_v(G)=\evc(G)$ and \\  \indent  
    $\evc_{uv}(G)=\min\{ \evc_u(G_u)+\mvc(G_v), \evc_v(G_v)+\mvc(G_u),
    \mvc_{uv}(G)+1 \}$.
\end{lemma}
\begin{proof}
Since $\evc_{uw}(G_u)\le k_u+1$ by Proposition~\ref{prop:occ_edge_vertices} and $\mvc_w(G_v)=k_v$ by Proposition~\ref{prop:evc-subset}, it follows by Lemma~\ref{lem:deg>2-122} that $\evc_u(G) \le \max\{\evc_{uw}(G_u)+\mvc_w(G_v)-1, \mvc_u(G_u)+\evc(G_v)\} =k_u+k_v=\evc(G)$ and by Lemma~\ref{lem:deg>2-3} Part~2, we get $\evc_u(G) \ge \evc(G)= k_u+k_v$, Hence, we get $\evc_u(G)=\evc(G)$. 
     Similarly, we get $\evc_v(G)=\evc(G)$.
            
     Now, we show that $\evc_{uv}(G)=\min\{\evc_u(G_u)+\mvc(G_v),
     \mvc(G_u)+\evc_v(G_v),\mvc_{uv}(G)+1\}$. 
     Since by Proposition~\ref{prop:evc-subset}, $\mvc_u(G_u)=k_u$ and $\mvc_v(G_v)=k_v$, by Lemma~\ref{lem:deg>2-2} Part~1, we have $\evc_{uv}(G)\le \min\{\evc_u(G_u)+\mvc(G_v),
     \mvc(G_u)+\evc_v(G_v),\mvc_{uv}(G)+1\}$.
      We know that $\evc_{uv}(G)\in \{\mvc_{uv}(G), \mvc_{uv}(G)+1\}$. So, it is now  enough to show that when $\evc_{uv}(G) = \mvc_{uv}(G)$, then $\min\{\evc_u(G_u)+\mvc(G_v),
     \mvc(G_u)+\evc_v(G_v)\}=\evc_{uv}(G)$. 
     
     Consider the case when $\evc_{uv}(G) = \mvc_{uv}(G)$. In this case, we will show that $\mvc_{uv}(G)=k_u+k_v$. Since we have already proved that $\evc_{u}(G)=\evc(G)=k_u+k_v$, by Proposition~\ref{prop:evc-subset}, we know that $\mvc_{uv}(G)\le k_u+k_v$.  Moreover, since $\evc_{uv}(G) \ge \evc(G)=k_u+k_v$, 
     we get $\evc_{uv}(G)= \mvc_{uv}(G) = k_u+k_v$.
       It remains to show that $\min\{\evc_u(G_u)+\mvc(G_v),\mvc(G_u)+\evc_v(G_v)\}=k_u+k_v$. Since $\mvc_{uv}(G)=k_u+k_v$, it can be easily seen that either $\mvc_{uw}(G_u)=k_u+1$ or $\mvc_{vw}(G_v)=k_v+1$. Without loss of generality, assume that $\mvc_{uw}(G_u)=k_u+1$. In this case, since $\evc_{uv}(G)=k_u+k_v$, by Lemma~\ref{lem:deg>2-3} Part~3, we get $\evc_v(G_v)=k_v$. Hence, we get $\min\{\evc_u(G_u)+\mvc(G_v),\mvc(G_u)+\evc_v(G_v)\}=k_u+k_v$.\qed
\end{proof}

\begin{lemma}\label{lem:case2algo-1}
If $\evc(G_u)=k_u$, $\evc(G_v)=k_v+1$ and $\mvc(G)=k_u+k_v-1$, then:\\
  $\evc_u(G)=\mvc_{uw}(G_u)+\evc(G_v)-1$, \\
  $\evc_v(G)=\min\{\mvc_v(G)+1,\evc_w(G_u)+\evc_{vw}(G_v)-1,\\ \indent \indent \indent \indent \max \{ \mvc_w(G_u)+\evc_{vw}(G_v)-1, \evc(G_u)+\mvc_v(G_v)\}\}$\\
  $\evc_{uv}(G)=\min \{\mvc_{uv}(G)+1, \evc_{uw}(G_u)+\evc_{vw}(G_v)-1, 
  \evc_{u}(G_u)+\max \{\evc_{vw}(G_v)-1, \mvc_v(G_v)\}\}$.
\end{lemma}

\begin{proof}
  Since $\evc(G_v)=k_v+1$, by Lemma~\ref{lem:algoevc-2}, it can be seen that $\evc(G)=k_u+k_v$.

  First, we will look at $\evc_u(G)$. By Lemma~\ref{lem:deg>2-3} Part~2, we get $\evc_u(G)\ge \mvc_{uw}(G_u)+\evc(G_v)-1$. Now, it is enough to show that 
  $\evc_u(G)\le \mvc_{uw}(G_u)+\evc(G_v)-1$.  
  We know that $\mvc_u(G) \in \{\mvc(G),\mvc(G)+1\}$. First, suppose $\mvc_u(G)=\mvc(G)$. In this case, it can be seen that $\mvc_{uw}(G_u)=k_u$ and hence, $\mvc_{uw}(G_u)+\evc(G_v)-1=k_u+k_v$. Moreover, by  Proposition~\ref{prop:evc-subset}, we know $\evc_u(G) \le \mvc_u(G)+1=k_u+k_v$. Hence, we are done.
  Next, suppose $\mvc_u(G)=\mvc(G)+1$.  In this case, by Lemma~\ref{lem:computemvcfix_gen-1}, we get $\mvc_{uw}(G_u)=k_u+1$. Hence, we get $\mvc_{uw}(G_u)+\evc(G_v)-1=k_u+k_v+1$. By Lemma~\ref{lem:deg>2-122}, we have $\evc_u(G) \le \evc(G)+1=k_u+k_v+1$. Hence, we are done.
  
  Now we will look at $\evc_v(G)$. 
  Let $Q= \min\{\mvc_v(G)+1,\evc_w(G_u)+\evc_{vw}(G_v)-1,
  \max \{ \mvc_w(G_u)+\evc_{vw}(G_v)-1, \evc(G_u)+\mvc_v(G_v)\}\}$,
  the expression for $\evc_v(G)$ to be proved. By Lemma~\ref{lem:deg>2-12}, we have $\evc_v(G)\le Q$. 
  Since $\evc_v(G) \in \{\mvc_v(G),\mvc_v(G)+1\}$ by Proposition~\ref{prop:evc-subset}, it is enough to show that if $\evc_v(G)=\mvc_v(G)$, then 
  $Q \le \mvc_v(G)$. Assume 
 $\evc_v(G)=\mvc_v(G)$. Since $\evc_v(G)\ge \evc(G)=k_u+k_v$
 and $\mvc_v(G) \le \evc(G)= k_u+k_v$, we have 
 $\evc_v(G)=\mvc_v(G)=k_u+k_v$. 
By Lemma~\ref{lem:deg>2-3} Part~2, we get $\evc_{vw}(G_v)\le k_v+1$.
  If $\evc_w(G_u)=k_u$, then we get $Q \le k_u+k_v$, as 
  required. If this was not the case, then there exists a 
  vertex $x \in V(G_u)$ with $\mvc_{xw}(G_u)=k_u+1$. If
  $\mvc_v(G_v)=k_v$, then we get $Q \le k_u+k_v$ 
  immediately. Hence, we assume $\mvc_v(G_v)=k_v+1$. Then, 
  for the vertex $x \in V(G_u)$ for which $\mvc_{xw}(G_u)=k_u+1$, 
  we get $\mvc_{xv}(G)>k_u+k_v$. 
   Since $\evc_v(G) \ge \mvc_{xv}(G)$ 
  by Proposition~\ref{prop:evc-subset}, this is a 
  contradiction. Thus, $Q \le k_u+k_v$.

  Now, we will consider $\evc_{uv}(G)$. Let $P=\min \{\mvc_{uv}(G)+1, \evc_{uw}(G_u)+\evc_{vw}(G_v)-1, 
  \evc_{u}(G_u)+\max \{\evc_{vw}(G_v)-1, \mvc_v(G_v)\}\}$,
  the expression for $\evc_{uv}(G)$ to be proved. 
  By Lemma~\ref{lem:deg>2-2}, we get $\evc_{uv}(G) \le P$. 
  Now, we will 
  show that $\evc_{uv}(G) \ge P$. Since $\evc_{uv}(G) \in \{\mvc_{uv}(G), \mvc_{uv}(G)+1\}$ by Proposition~\ref{prop:evc-subset}, it is enough to show that if $\evc_{uv}(G)=\mvc_{uv}(G)$, then 
  $P \le \mvc_{uv}(G)$. Assume 
   $\evc_{uv}(G)=\mvc_{uv}(G)$. Since 
  $\evc_{uv}(G)\ge \evc(G)= k_u+k_v$
 and $\mvc_{uv}(G) \le k_u+k_v$, we know that 
 $\evc_{uv}(G)=\mvc_{uv}(G)=k_u+k_v$. By Lemma~\ref{lem:deg>2-3} Part~3, we get $\evc_{vw}(G_v)=k_v+1$. 
 If $\evc_{uw}(G_u)=k_u$, then it can be seen that $P\le k_u+k_v$. Hence, assume 
 $\evc_{uw}(G_u)=k_u+1$. Then, by Lemma~\ref{lem:deg>2-3} Part~3, we get $\mvc_v(G_v)=k_v$. Further, by Lemma~\ref{lem:deg>2-3} Part~3, it can be seen that $\mvc_{uw}(G_u)=k_u$. This implies $\mvc_{vw}(G_v)=k_v+1$ by Lemma~\ref{lem:computemvcfix_gen-1}. Hence, by Lemma~\ref{lem:deg>2-3}, we get $\evc_u(G_u)=k_u$. Therefore, we get $P \le k_u+k_v=\evc_{uv}(G)$.
  \end{proof}

  \begin{lemma}\label{lem:case2algo-2}
If $\evc(G_u)=k_u$, $\evc(G_v)=k_v+1$ and $\mvc(G)=evc(G)=k_u+k_v$, then:\\ $\evc_u(G)=\evc_u(G_u)+\evc_w(G_v)-1$,\\
$\evc_v(G)=\min \{\mvc_v(G)+1, \evc_w(G_u)+\evc_{vw}(G_v)-1,\\ \indent \indent \indent \indent \max \{ \mvc_w(G_u)+\evc_{vw}(G_v)-1, \evc(G_u)+\mvc_v(G_v)\}\}$\\
 $\evc_{uv}(G)= \min\{\evc(G)+1, \max\{\evc_u(G_u)+\mvc_v(G_v),\evc_u(G_u)+\evc_{vw}(G_v)-1\}\}$.
\end{lemma}

\begin{proof}

  First, we will show that 
  $\evc_u(G)=\evc_u(G_u)+\evc_w(G_v)-1$. Since $\evc(G)=k_u+k_v$, 
  by Lemma~\ref{lem:deg>2-3} Part~1, $\evc_w(G_v)\le k_v+1$, which implies $\evc_w(G_v)= k_v+1$. Further, by Proposition~\ref{prop:occ_edge_vertices}, $\evc_u(G) \in\{k_u+k_v, k_u+k_v+1\}$. Now, it is enough to  show that if $\evc_u(G)=\evc(G)=k_u+k_v$, then $\evc_u(G_u)=k_u$ and if $\evc_u(G)=\evc(G)+1=k_u+k_v+1$, then $\evc_u(G_u)=k_u+1$. First, suppose $\evc_u(G)=k_u+k_v$. By Lemma~\ref{lem:deg>2-3} Part~2, we get $\mvc_{uw}(G_u)=k_u$ and by Lemma~\ref{lem:computemvc2}, we get $\mvc_w(G_v)\ne k_v$. Hence, we have $\mvc_v(G_v)=k_v$.  
  For contradiction, suppose $\evc_u(G_u)=k_u+1$. Then, by Proposition~\ref{prop:evc-subset},  there exists a vertex $x \in V(G_u)$, such that 
   $\mvc_{xu}(G_u)= k_u+1$. Since $\mvc_w(G_v)=k_v+1$, this implies $\mvc_{xu}(G)\ge k_u+k_v+1$. By Proposition~\ref{prop:evc-subset}, this contradicts the value of $\evc_u(G)$. Hence, we get $\evc_u(G_u)=k_u$ as required. 
   Next, suppose $\evc_u(G)=k_u+k_v+1$. For contradiction, suppose $\evc_u(G_u)=k_u$. Then, by Proposition~\ref{prop:evc-subset}, we get $\mvc_{uw}(G_u)=k_u$.
   By Lemma~\ref{lem:deg>2-122}, we have
   $\evc_u(G) \le \max\{\evc_u(G_u)+\mvc(G_v), \mvc_{uw}(G_u)+\evc_w(G_v)-1\}=k_u+k_v$, which is a contradiction. Hence, $\evc_u(G_u)=k_u+1$ and we are done.
  
  The proof of the expression for $\evc_v(G)$ is  exactly the same as that in the proof of Lemma~\ref{lem:case2algo-1}.
  
    Now we will look at $\evc_{uv}(G)$. 
  Let $P=  \max\{\evc_u(G_u)+\mvc_v(G_v),\evc_u(G_u)+\evc_{vw}(G_v)-1\}$. 
  We need to show that $\evc_{uv}(G)=\min\{\evc(G)+1, P\}$. By Lemma~\ref{lem:deg>2-2}, we get $\evc_{uv}(G)\le \min\{\evc(G)+1, P\}$. Now, 
  we will show that $\evc_{uv}(G)\ge \min\{\evc(G)+1, P\}$. If 
  $\evc_{uv}(G)=\evc(G)+1$, we are done. Therefore, consider the case when 
  $\evc_{uv}(G)=\evc(G)=k_u+k_v$.
   In this case, we show that  
  $P \le k_u+k_v$. 
  Since $\mvc(G)=k_u+k_v$ and $\mvc_w(G_u)=k_u$ by  Proposition~\ref{prop:evc-subset}, we get $\mvc_{vw}(G_v)=k_v+1$ by Lemma~\ref{lem:computemvc2}. Hence, by Lemma~\ref{lem:deg>2-3} Part~3, we get $\evc_u(G_u)=k_u$. Moreover, we also get $\evc_{vw}(G_v)=k_v+1$ and $\mvc_{uw}(G_u)=k_u$ by Lemma~\ref{lem:deg>2-3} Part~3. Therefore, since $\mvc(G)=k_u+k_v$, we  know $\mvc_w(G_v)\ne k_v$. This implies $\mvc_v(G_v)=k_v$. Thus, we get $P\le k_u+k_v$.
  \end{proof}
\begin{lemma}\label{lem:case2algo-3}
If $\evc(G_u)=k_u$, $\evc(G_v)=k_v+1$ and  $\evc(G)=k_u+k_v+1$, then\\ $\evc_u(G)=\evc_v(G)=\evc(G)$ and
$\evc_{uv}(G)=\min \{\mvc_{uv}(G)+1,\evc(G_u)+\evc_{v}(G_v)\} $.
\end{lemma}
\begin{proof}
  First, we will show that $\evc_u(G)=\evc(G)$. 
  Since we know that $\evc_u(G)\ge \evc(G)$ by Lemma~\ref{lem:deg>2-3} Part~2, it is enough to
  show that $\evc_u(G)\le \evc(G)$. 
  For any vertex $x \in V(G_u)$, we have $\mvc_x(G_u)=k_u$ by Proposition~\ref{prop:evc-subset} and 
  hence $\mvc_{xu}(G_u) \le k_u+1$. This implies 
  $\mvc_{xu}(G) \le k_u+k_v+1=\evc(G)$. Similarly, 
  for any vertex $y \in V(G_v)$, we have $\mvc_y(G_v)\le k_v+1$.
  Hence, $\mvc_{yu}(G)\le k_u+k_v+1$.
  Therefore, by Proposition~\ref{prop:evc-subset},
   $\evc_u(G) \le k_u+k_v+1$. Hence, $\evc_u(G)=\evc(G)$.
   
   Next, we will show that $\evc_v(G)=\evc(G)$. 
  Since we know that $\evc_v(G)\ge \evc(G)$ by Lemma~\ref{lem:deg>2-3} Part~2, it is enough 
  to show that $\evc_v(G)\le \evc(G)$. 
  For any vertex $x \in V(G_u)$, we have $\mvc_x(G_u)=k_u$ by Proposition~\ref{prop:evc-subset} and 
  hence $\mvc_{xv}(G) \le k_u+k_v+1=\evc(G)$.
   Similarly, for any vertex $y \in V(G_v)$, 
  $\mvc_{yvw}(G_v)\le k_v+2$. Since $\mvc_w(G_u)=k_u$, this gives 
  $\mvc_{yv}(G)\le k_u+k_v+1$. Therefore, by Proposition~\ref{prop:evc-subset},
   $\evc_v(G) \le k_u+k_v+1$. Hence, $\evc_v(G)=\evc(G)$. 
   
    Now, we will show that $\evc_{uv}(G)=\min \{\mvc_{uv}(G)+1,\evc(G_u)+\evc_{v}(G_v)\}$. By Lemma~\ref{lem:deg>2-2}, we get $\evc_{uv}(G)\le \min \{\mvc_{uv}(G)+1,\evc(G_u)+\evc_{v}(G_v)\}$.  Next, we show that $\evc_{uv}(G)\ge \min \{\mvc_{uv}(G)+1,\evc(G_u)+
    \evc_{v}(G_v)\}$. If $\mvc_{uv}(G)=k_u+k_v$ or $\evc_v(G_v)=k_v+1$, this is
    trivial. So, we are left with the case when $\mvc_{uv}(G)=k_u+k_v+1$
    and $\evc_v(G_v)=k_v+2$. 
    For contradiction, suppose $\evc_{uv}(G)\le k_u+k_v+1$.    
    By Proposition~\ref{prop:evc-subset}, 
    there exists a vertex 
    $x \in V(G_v)$ such that $\mvc_{xv}(G_v)=k_v+2$ and
    this implies $\mvc_{xuv}(G)\ge k_u+k_v+1$ and $\mvc_{uw}(G_u)=k_u$. 
    Since $\mvc_{uv}(G)=k_u+k_v+1$, we get $\mvc_v(G_v)=k_v+1$ 
    and $\mvc_w(G_v)=k_v$. This implies $\mvc(G)=k_u+k_v-1$ and 
    $\mvc_{uv}(G)\le k_u+k_v$, a contradiction. Therefore, 
    $\evc_{uv}(G)\ge k_u+k_v+2$, as required. 
    Hence, $\evc_{uv}(G)= \min \{\mvc_{uv}(G)+1,\evc(G_u)+\evc_{v}(G_v)\}$. 
\end{proof}\qed

\begin{lemma}\label{lem:case3algo}
 If $\evc(G_u)=k_u+1$ and 
  $\evc(G_v)=k_v+1$, then \\ \indent$\evc_u(G)=\min \{ \mvc_u(G)+1,\evc_u(G_u)+\evc_w(G_v)-1\}$ and \\
  \indent $\evc_v(G)=\min \{ \mvc_v(G)+1,\evc_v(G_v)+\evc_w(G_u)-1\}$.
\end{lemma}

\begin{proof}
 Since the statement of the lemma is symmetric with respect to $u$ and $v$, we will only prove the
 expression for $\evc_u(G)$.
 We first show that $\evc_u(G) \le \min \{ \mvc_u(G)+1,\evc_u(G_u)+\evc_w(G_v)-1\}$. Since $\evc_u(G) \le \mvc_u(G)+1$ by Proposition~\ref{prop:evc-subset}, 
 it is enough to show that $\evc_u(G) \le \evc_u(G_u)+\evc_w(G_v)-1$. By Proposition~\ref{prop:mvc_ux}, we get $\max_{x \in V(G_u)}\mvc_{ux}(G) \le \evc_u(G_u)+\mvc(G_v) \le \evc_u(G_u)+\evc_w(G_v)-1$ and $\max_{x \in V(G_v)}\mvc_{ux}(G) \le \mvc_{uw}(G_u)+\evc_w(G_v)-1 \le \evc_u(G_u)+\evc_w(G_v)-1$. 
 Hence, by Proposition~\ref{prop:evc-subset}, we get $\evc_u(G)\le \evc_u(G_u)+\evc_w(G_v)-1$. 
 
  Next, we need to show that $\evc_u(G)\ge \min \{ \mvc_u(G)+1,\evc_u(G_u)+\evc_w(G_v)-1\}$. By Proposition~\ref{prop:evc-subset}, it is enough to show that when $\evc_u(G_u)=\mvc_u(G)$, then $\evc_u(G)\ge \evc_u(G_u)+\evc_w(G_v)-1$.
  Suppose $\evc_u(G)=\mvc_u(G)$. Since $\mvc_u(G) \le \evc(G) \le \evc_u(G)$, we have $\mvc_u(G) = \evc(G) = \evc_u(G)$. By Observation~\ref{obs:computemvc1} and Lemma~\ref{lem:algoevc-3}, we have $\evc(G) \in \{k_u+k_v,k_u+k_v+1\}$.  First, suppose $\evc_u(G)=k_u+k_v$. We will show a contradiction in this case. By Lemma~\ref{lem:deg>2-3} Part~2, we get $\mvc_w(G_v)=k_v$ and $\mvc_{uw}(G_u)=k_u$. This implies $\mvc_u(G)=k_u+k_v-1$, which is a contradiction.
 
 Next, suppose $\evc_u(G)=k_u+k_v+1$. In this case, it is enough to show that $\evc_u(G_u)=k_u+1$
 and $\evc_w(G_v)=k_v+1$.
 We have $\mvc_{uw}(G_u) \le k_u+1$. If $\mvc_w(G_v)=k_v$, we get $\mvc_u(G)=k_u+k_v$,
 a contradiction. Therefore, $\mvc_w(G_v)=k_v+1$. 
 By Lemma~\ref{lem:deg>2-3} Part~2, we get $\evc_u(G_u)=k_u+1$. If $\mvc_u(G_u)=k_u$, we get $\mvc_u(G)=k_u+k_v$, a contradiction. Therefore, $\mvc_u(G_u)=k_u+1$.
 Consider any vertex $x \in V(G_v)$. By Proposition~\ref{prop:evc-subset}, we have $\mvc_{xu}(G) \le k_u+k_v+1$. Since, $\mvc_u(G_u)=k_u+1$, it can
 be verified that 
 $\mvc_{xw}(G_v) \le k_v+1$. Hence, by Proposition~\ref{prop:evc-subset}, $\evc_w(G_v)=k_v+1$, as required.
\end{proof}

\begin{lemma}\label{lem:case4algo}
 If $\evc(G_u)=k_u+1$ and $\evc(G_v)=k_v+1$, then: 
  \begin{enumerate}
   \item If $\mvc_{uv}(G) \le k_u+k_v$ then $\evc_{uv}(G)=\mvc_{uv}(G)+1$.
   \item Otherwise, $\evc_{uv}(G)=min\{\mvc_{uv}(G)+1, max\{P_1, P_2\}\}$, where \\   $P_1=\min\{\evc_{uw}(G_u)+\mvc_{vw}(G_v)-1, \evc_u(G_u)+\mvc_v(G_v)\}$ and\\ $P_2=\min\{\mvc_{uw}(G_u)+\evc_{vw}(G_v)-1, \mvc_u(G_u)+\evc_v(G_v)\}$.
     \end{enumerate}   
\end{lemma}
\begin{proof} 
\begin{enumerate}
 \item Suppose $\mvc_{uv}(G) \le k_u+k_v$. 
 By Proposition~\ref{prop:evc-subset}, $\evc_{uv}(G) \in \{\mvc_{uv}(G),\mvc_{uv}(G)+1\}$. Therefore it is enough to show that $\evc_{uv}(G) \ne \mvc_{uv}(G)$. For contradiction, suppose this was not true. By Lemma~\ref{lem:deg>2-3} Part~1, we get $\evc(G) \ge k_u+k_v$. Therefore, $\evc_{uv}(G) \ge k_u+k_v$.  Hence, we may assume that $\mvc_{uv}(G)=\evc_{uv}(G)=k_u+k_v$. By Lemma~\ref{lem:deg>2-3} Part~3, we get $\mvc_{uw}(G_u)=k_u$ and $\mvc_{vw}(G_v)=k_v$. This gives $\mvc_{uv}(G)=k_u+k_v-1$, a contradiction. Hence, $\evc_{uv}(G)=\mvc_{uv}(G)+1$.
\item Suppose $\mvc_{uv}(G) = k_u + k_v + 1$. By Lemma~\ref{lem:deg>2-2}, we know $\evc_{uv}(G) \le \min\{\mvc_{uv}(G)+1, \max(P_1, P_2)\}$.
Now, we will show that $\evc_{uv}(G) \ge \min(\mvc_{uv}(G)+1, \max(P_1, P_2))$. Recall that we are in a case when $\mvc_{uv}(G)= k_u+k_v+1$ and by Proposition~\ref{prop:evc-subset}
$\evc_{uv}(G) \in \{\mvc_{uv}(G), \mvc_{uv}(G)+1\}$. If $\evc_{uv}(G) = \mvc_{uv}(G)+1$, we are done. Therefore we are left with the case where
$\evc_{uv}(G)=\mvc_{uv}(G)=k_u+k_v+1$. In this case $\mvc_{uw}(G_u)=k_u+1$ and $\mvc_{vw}(G_v)=k_v+1$.
It is enough to show that $\max(P_1,P_2) \le k_u+k_v+1$. For contradiction, suppose $P_1 > k_u+k_v+1$.
This implies $\evc_{uw}(G_u)=k_u+2$. By Lemma~\ref{lem:deg>2-3} Part~3,  we get $\mvc_v(G_v)=k_v$, $\evc_v(G_v)=k_v+1$ and $\evc_u(G_u)=k_u+1$. From this, we get $P_1 \le k_u+k_v+1$, a contradiction.

Symmetrically, we get contradiction when $P_2 > k_u+k_v+1$. Therefore, $\max(P_1,P_2) \le k_u+k_v+1$.\qed
\end{enumerate}
\end{proof}
From Lemma~\ref{lem:algoevc-1}-\ref{lem:case4algo}, the following theorem is immediate. 
 
 \begin{theorem}\label{thm:vc_evc_general-2}
 Let $G$ be a maximal outerplanar graph and $uv$ be an edge on the outer face of $G$ such that $deg_G(u)>2$ and $deg_G(v)>2$.  Let $w=\Delta(uv)$.
 \begin{enumerate}
    \item Given $\mathscr{M}(G_u, uw)$, $\mathscr{E}(G_u, uw)$, $\mathscr{M}(G_v, vw)$ and $\mathscr{E}(G_v, vw)$, 
 it is possible to compute $\evc(G)$ in constant time.
 \item  Given $\evc(G)$, $\mathscr{M}(G, uv)$, $\mathscr{M}(G_u, uw)$, $\mathscr{E}(G_u, uw)$, $\mathscr{M}(G_v, vw)$ and $\mathscr{E}(G_v, vw)$, it is possible to compute the remaining evc parameters of $G$ with respect to $uv$ in constant time.
  \end{enumerate}
 \end{theorem}
\section{A linear time algorithm to compute $\evc$ number}
In this section, we formulate a divide and conquer algorithm that takes a pair $(G, uv)$ as input where $G$ is a maximal outerplanar graph and $uv$ is an edge  on its outer face and recursively computes $\mathscr{M}(G,uv)$  and $\mathscr{E}(G,uv)$. The case when $G$ is a triangle forms the base case of the recursion and it is handled using Observation~\ref{obs:basecase}. Let $w=\Delta(uv)$. If $deg_G(u)>2$ and  $deg_G(v)>2$, then the recursion works on $(G_u, uw)$ and $(G_v, vw)$ using Theorem~\ref{thm:vc_evc_general-1} and  Theorem~\ref{thm:vc_evc_general-2}. Otherwise, suppose $u'$ is the degree-2 vertex among $u$ and $v$ and $v'$ is the other one. In this case,  the recursion works on $(G_{v'}, v'w)$ using Theorem~\ref{thm:vc_evc_deg2-1} and Theorem~\ref{thm:vc_evc_deg2-2}. 

A high level overview of our method is given in Algorithm~\ref{algo:alg1}. To obtain the linear time guarantee, we need to ensure that the time spent in steps other than the recursive calls is $O(1)$.
Theorems~\ref{thm:vc_evc_deg2-1}-\ref{thm:vc_evc_general-2} guarantee that lines $8-11$ and lines $17-20$ work in constant time, forming the most crucial part of our algorithm. However, we will have to avoid the explicit computation of the subgraphs and passing them as explicit parameters in the recursive calls using a refinement of  Algorithm~\ref{algo:alg1}.
For this refinement, we maintain a 
Vertex Edge Face Adjacency List data structure using which the following are possible: 
\begin{enumerate}                                                            \item For any edge $e$, we can traverse through the internal faces on which the edge lies in $O(1)$ time.
\item For any internal face $f$, we can traverse through the edges on the boundary of $f$ in $O(1)$ time.                                                                                       \end{enumerate} 
The details of the data structure is given in Fig.~\ref{fig:node} in the appendix. An example for Vertex Edge Face Adjacency List of a maximal outerplanar graph on 
four vertices is given in Fig.~\ref{fig:Example} in the appendix. Vertex Edge Face Adjacency List of an outerplanar graph  can be produced by modifying the algorithm suggested by N.~Chiba and T.~Nishizeki \cite{chiba1985arboricity} for listing all the 
triangles of a planar graph in linear time. An algorithm for producing the Vertex Edge Face Adjacency List is given in Algorithm~\ref{algo:alg3} in the appendix. 
 
By using the \textit{visit} field of faces in the Vertex Edge Face Adjacency List data structure, we avoid passing the subgraph as an explicit parameter in
recursive calls, maintaining the following invariants:
\begin{enumerate}
 \item The edge $e$ passed as parameter to the recursion is present in exactly one 
unvisited face $f$.
\item The $e$-segment of $G$ that contains the unvisited face $f$ is precisely the subgraph on which the recursive call in Algorithm~\ref{algo:alg1} 
would have been made. 
\end{enumerate}
 
Further, line numbers $2$, $3$, $5$ and $6$ of Algorithm~\ref{algo:alg1} can also be implemented in linear time using the new data structure. 
\begin{algorithm}[H]
  \caption{To compute $\mathscr{M}(G,uv)$ and $\mathscr{E}(G,uv)$ of a maximal outerplanar graph $G$ with $uv$ an edge on its outer face.}
   \textbf{Inputs}: A maximal outerplanar graph $G$ with at least three vertices, an edge $uv$ on the outer face of $G$.\\
 \textbf{Outputs}: $\mathscr{M}(G,uv)$ and $\mathscr{E}(G,uv)$.
  \label{algo:alg1}
  \begin{algorithmic}[1]
   \Procedure{EVC\_Parameters} {$G$, $uv$}
\State{$w \gets $ $\Delta(uv)$.}
\If{$G$ is a triangle}
    \State $\mvc(G)$=$\evc(G)$=$\mvc_u(G)$=$\mvc_v(G)$=$\mvc_{uv}(G)$=$2$, $\evc_u(G)$=$\evc_v(G)$=$2$, $\evc_{uv}(G)$=$3$.

\ElsIf{One end point of $uv$ is degree-2}
     \State $u'$=degree-2 end point of $uv$, $v'$=higher degree end point of $uv$.
       \State Recursively call EVC\_Parameters$(G_{v'},v'w)$.
       \State {Compute $\mvc(G)$ using Theorem~\ref{thm:vc_evc_deg2-1} (Part 1).}
       \State {Compute $\mvc_u(G)$, $\mvc_v(G)$ and $\mvc_{uv}(G)$ using  Theorem~\ref{thm:vc_evc_deg2-1} (Part 2).} 
       \State {Compute $\evc(G)$ using Theorem~\ref{thm:vc_evc_deg2-2} (Part 1).} 
        \State {Compute $\evc_{uv}(G)$ in constant using Theorem~\ref{thm:vc_evc_deg2-2} (Part 2).} 

\Else 

\State{$G_u \gets$ $uv$-segment of $G$ containing the vertex $u$}
\State{$G_v \gets$ $uv$-segment of $G$ containing the vertex $v$}
 \State  Recursively call EVC\_Parameters$(G_u,uw)$.
 \State Recursively call EVC\_Parameters$(G_v,vw)$.
 \State {Compute $\mvc(G)$ using Theorem~\ref{thm:vc_evc_general-1} (Part 1).} 
 \State {Compute $\mvc_u(G)$, $\mvc_v(G)$, $\mvc_{uv}(G)$ using  Theorem~\ref{thm:vc_evc_general-1} (Part 2).} 
 \State {Compute $\evc(G)$ using Theorem~\ref{thm:vc_evc_general-2} (Part 1).} 
  \State {Compute $\evc_u(G)$, $\evc_v(G)$, $\evc_{uv}(G)$ using  Theorem~\ref{thm:vc_evc_general-2} (Part 2).} 
         
  \EndIf
 \State\Return($\mathscr{M}(G,uv)$ and $\mathscr{E}(G,uv)$)
\EndProcedure
\end{algorithmic}
\end{algorithm}

Algorithm~\ref{algo:alg2} (in the appendix) gives a linear time refinement of Algorithm~\ref{algo:alg1}. The details of the correspondence between the two algorithms are given in the appendix.

\section{Conclusion}\vspace{-.15cm}
This paper presents a linear time algorithm for computing the eternal vertex cover number of maximal outerplanar graphs,
lowering the best known upper bound to the complexity of the problem from quadratic \cite{BP2021} time to linear.    The techniques
presented in the paper make crucial use of the planarity of the underlying graph to yield a divide and conquer algorithm for 
computing the $\evc$ number of a maximal outerplanar graph. Attempts to generalize
the techniques to maximal planar graphs may not be successful due to the known NP hardness result on the evc computation of biconnected internally triangulated planar graphs \cite{BABU2021}. However, 
the complexity status of the problem of computing the $\evc$ number of outerplanar graphs is open and may be attempted using the techniques developed in this work.

\clearpage
\section{Appendix - Refinement of Algorithm~\ref{algo:alg1} to Work in Linear Time} \label{sec:appendix}

\subsection{Vertex Edge Face Adjacency List} \label{sec:VEFAl}

To make Algorithm~\ref{algo:alg1} work in linear time, we should avoid the explicit passing of the subgraphs $G_u$ and $G_v$ as parameters
for the recursive calls. Instead, we do the following.

Using an algorithm suggested by N.~Chiba and T.~Nishizeki \cite{chiba1985arboricity}, in linear time all the 
triangles of a planar graph can be listed/ printed. Let $G$ be a maximal outerplanar graph with a fixed outerplanar embedding.
When $G$ is given as input to the algorithm of Chiba and Nishizeki, it will precisely list 
all the internal faces of the graph.
By modifying the algorithm proposed in \cite{chiba1985arboricity}, we produce a global data structure named 
\textbf{Vertex Edge Face Adjacency List}.
The details of the data structure is given in Fig.~\ref{fig:node}. The algorithm for producing this data structure is given in Algorithm~\ref{algo:alg3} and will be discussed later.
   
 \begin{figure}
  \centering
  \includegraphics[width=\textwidth]{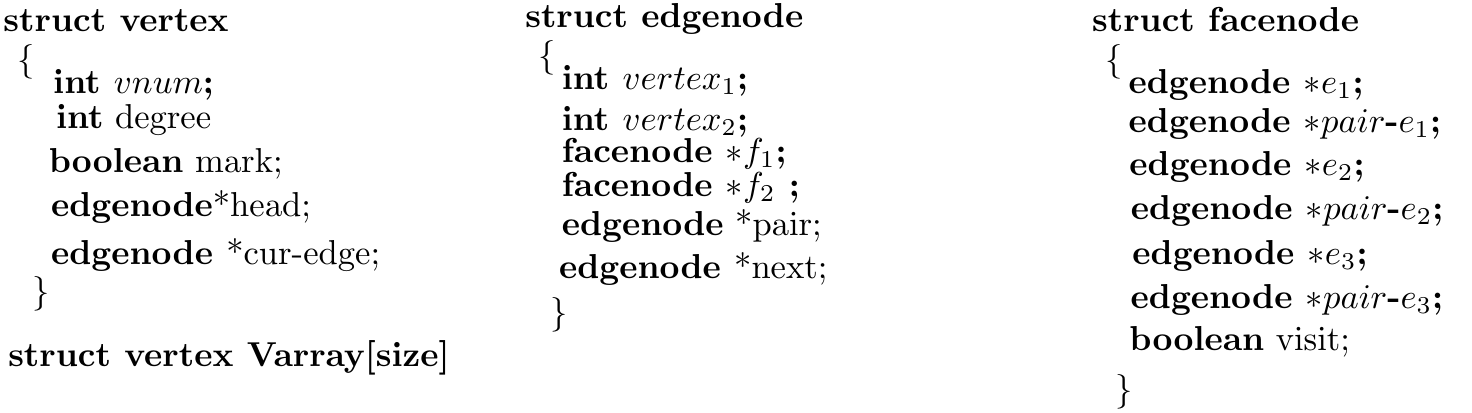}
  \caption{Details of the Vertex Edge Face Adjacency List data structure.}
  \label{fig:node}
 \end{figure}
 
In a Vertex Edge Face Adjacency List data structure, for any edge $v_iv_j$, there are links between the edge node $v_iv_j$ and the edge node $v_jv_i$. For each face node representing a face $f$ in $G$, there are links to the edge nodes corresponding to the bounding edges of $f$ and vice versa. Remember that each edge of $G$ has 2 edge nodes corresponding to it. Hence, for each face node $f$, there are 6 pointers to edge nodes. For each face node, there is a visit field which is initialized
to \textit{unvisited}. Each vertex node has a mark field, which is initialized to \textit{unmarked}, the cur-edge field which is initialized to \textit{null} and a pointer to the beginning of the adjacency list of that vertex.

The Vertex Edge Face Adjacency List of a maximal outerplanar graph on four vertices is given in Fig.~\ref{fig:Example}.




 \subsection{Refined algorithm and its correspondence with Algorithm~\ref{algo:alg1}}

We will first assume that Vertex Edge Face Adjacency List of the graph has been already produced and give the remaining details of the linear time implementation. Algorithm~\ref{algo:alg2} gives a linear time version of Algorithm~\ref{algo:alg1}. To compute the $\evc$ number of a
 maximal outerplanar graph, 
we pass an edge $uv$ on the outer face of the graph as input to the algorithm.
Initially, the visit field of all faces are set to \textit{unvisited}. Note that,
the edge $uv$ is present in exactly one unvisited face $uvw$ initially. 
On entering the algorithm, the face $uvw$ is marked as \textit{visited}.
Once this marking is done, the edges $uw$ and $vw$ can be in at most one unvisited face each.
Note that if $uw$ is not in any unvisited
face, then $deg_{G}(u)=2$. Otherwise $deg_{G}(u) > 2$. Similarly we can check if $deg_{G}(v)=2$ or not from the number of unvisited faces containing the edge $vw$.

\begin{algorithm}
  \caption{For computing $\mathscr{M}(G,uv)$ and $\mathscr{E}(G,uv)$ of a maximal outerplanar graph $G$ with $uv$ an edge on its outer face in linear time}
   \textbf{Inputs}: An edge $uv$ on the outer face of the maximal outerplanar graph $G$. \\
 \textbf{Output}: $\mathscr{M}(G,uv)$ and $\mathscr{E}(G,uv)$.\\
  \textbf{Assumption}: The Vertex Edge Face Adjacency List is maintained globally.
   \label{algo:alg2}
  \begin{algorithmic}[1]
   \Procedure{EVC\_Parameters} {}
     \State{Identify the unvisited face $uvw$ in which the edge $uv$ lie.} \\ \Comment{Possible in constant time by Vertex Edge Face Adjacency List}
     \State{ Set the visit field of the face $uvw$ to visited.}
\If{neither edge $uw$ nor edge $vw$ is in any unvisited faces} \\
\Comment{Possible to check in constant time using Vertex Edge Face Adjacency List}
   \State $\mvc(G)$=$\evc(G)$=$\mvc_u(G)$=$\mvc_v(G)$=$\mvc_{uv}(G)$=$2$, $\evc_u(G)$=$\evc_v(G)$=$2$, $\evc_{uv}(G)$=$3$.

\ElsIf{exactly one among $uw$ and $vw$ lie in an unvisited face} \\ \Comment{Possible to check in constant time using the Vertex Edge Face Adjacency List}
     \State  $e$ be the edge among $uw$ and $vw$ that lie in an unvisited face.
       \State $t \longleftarrow$ EVC\_Parameters$(e)$.
       \State {Compute $\mvc(G)$ in constant time using Theorem~\ref{thm:vc_evc_deg2-1} (part 1).}
       \State {Compute $\mvc_u(G)$, $\mvc_v(G)$ and $\mvc_{uv}(G)$ in constant time  using Theorem~\ref{thm:vc_evc_deg2-1} (part 2).} 
       \State {Compute $\evc(G)$ in constant time using Theorem~\ref{thm:vc_evc_deg2-2} (part 1).} 
        \State {Compute $\evc_{uv}(G)$ in constant using Theorem~\ref{thm:vc_evc_deg2-2} (part 2).} 

\Else 

 \State  $t_1 \longleftarrow$ EVC\_Parameters$(uw)$.
 \State $t_2 \longleftarrow$  EVC\_Parameters$(vw)$.
 \State {Compute $\mvc(G)$ in constant time using Theorem~\ref{thm:vc_evc_general-1} (part 1).} 
 \State {Compute $\mvc_u(G)$, $\mvc_v(G)$, $\mvc_{uv}(G)$ in constant time using  Theorem~\ref{thm:vc_evc_general-1} (part 2).} 
 \State {Compute $\evc(G)$ in constant time using Theorem~\ref{thm:vc_evc_general-2} (part 1).} 
  \State {Compute $\evc_u(G)$, $\evc_v(G)$, $\evc_{uv}(G)$ in constant time using  Theorem~\ref{thm:vc_evc_general-2} (part 2).} 
         
  \EndIf
 \State\Return($\mathscr{M}(G,uv)$ and $\mathscr{E}(G,uv)$)
\EndProcedure
\end{algorithmic}
\end{algorithm}


In subsequent recursive calls, the following invariants will be maintained: \\(1) the edge $e$ passed as parameter to the recursion is present in exactly one 
unvisited face $f$. \\
(2) The $e$-segment of $G$ that contains the unvisited face $f$ is precisely the subgraph on which the recursive call in Algorithm~\ref{algo:alg1} 
would have been made.\\
By maintaining the visit field of faces, we avoid passing the subgraph as an explicit parameter in recursive calls in 
Algorithm~\ref{algo:alg2}.

From these observations, we show that Algorithm~\ref{algo:alg2} is a linear time implementation of Algorithm~\ref{algo:alg1}.
The inputs of the function EVC\_Parameters  of Algorithm~\ref{algo:alg1} are a maximal outerplanar graph $G$ and an edge $uv$ on the outer face.
The algorithm is divided into following three cases:
\begin{enumerate}
 \item The input graph is a triangle (degrees of both the end vertices of the edge $uv$ are $2$). Refer to line number 3 of Algorithm~\ref{algo:alg1}. 
 \item The input graph is not a triangle and exactly one end vertex of the input edge $uv$ is of degree $2$. 
 Refer to line number 5 of Algorithm~\ref{algo:alg1}. 
 \item The input graph is not a triangle and both the end vertices of the input edge $uv$ has degree greater than $2$. 
 Refer to line number 12 of Algorithm~\ref{algo:alg1}.. 
\end{enumerate}
The input to the function EVC\_Parameters of Algorithm~\ref{algo:alg2} is an edge $uv$ on the outer face of the input graph.
For the input edge $uv$, let $uvw$ be the unvisited face in which $uv$ lies.
The three cases of Algorithm~\ref{algo:alg2} corresponding to the three cases of Algorithm~\ref{algo:alg1} are as follows:
\begin{enumerate}
 \item Both edges $uw$ and $vw$ lie in no unvisited face and thereby $uvw$ is a triangle.
 Refer to line number 5 of Algorithm~\ref{algo:alg2}. 
 \item Exactly one among the edges $uw$ and $vw$ lie in an unvisited face. Refer to line number 8 of Algorithm~\ref{algo:alg2}. 
 \item Both $uw$ and $vw$ lie in one unvisited face each. Refer to line number 16 of Algorithm~\ref{algo:alg2}. 
\end{enumerate}
In line number 2 of Algorithm~\ref{algo:alg1}, we find the unique common neighbor $w$ of $u$ and $v$ in the input graph.
Instead of this step, line number 2
of Algorithm~\ref{algo:alg2}, identifies $w$ as the third vertex in the unique unvisited face containing the edge $uv$. 
After this the face $uvw$ is marked as visited in line number 4 of Algorithm~\ref{algo:alg2} and subsequently the edges $uw$ and $vw$ are in at most
one unvisited face each. 

Now, we show that line number 3 of Algorithm~\ref{algo:alg1} is equivalent to line number 5 of Algorithm~\ref{algo:alg2}.
In line number 3 of Algorithm~\ref{algo:alg1}, we check whether the input graph is a triangle, which is the base case. 
Equivalently, by the invariants stated above, in Algorithm~\ref{algo:alg2}, it is enough to check if the 
$uv$-segment of $G$ containing the face $uvw$ is a triangle.
In   
line 5 of Algorithm~\ref{algo:alg2}, we do the following in constant time using vertex edge face adjacency list:
We check if $uw$ or $vw$ lie in any unvisited face. If neither $uw$ nor $vw$ lie in an unvisited face,
then we can infer that $uv$-segment of $G$ containing the face $uvw$ is a triangle. 

In line 5 of Algorithm~\ref{algo:alg1}, we check whether exactly one end vertex of the edge $uv$ is of degree $2$. 
In line number 8 of Algorithm~\ref{algo:alg2}, if $uw$ (respectively, $vw$) is in any univisited
face, then we can infer that the  degree of the vertex $u$ (respectively, $v$) is greater than $2$ in the $uv$-segment of 
$G$ that contains the face $uvw$. Otherwise, the degree of $u$ (respectively, $v$) in the $uv$-segment of $G$
that contains the face $uvw$ is equal to $2$.

In line number 7 of Algorithm~\ref{algo:alg1}, we recursively call the function EVC\_Parameters with two input parameters: 
(1) the graph obtained by deleting the degree-2 endpoint $u'$ of the edge $uv$ from the input graph and 
(2) the edge $v'w$ bounding the face $uvw$, where degree of $v'$ is greater than $2$ in the input graph.
Since in Algorithm~\ref{algo:alg2}, $uvw$ is already marked as visited in line number 4,
 it is enough to invoke the recursive call with the edge $e$ as parameter, where $e$ is the bounding edge 
of the face $uvw$ with one unvisited face adjacent to it.
This is achieved in line number 12 of Algorithm~\ref{algo:alg2}, in which
we recursively call the function EVC\_Parameters with the edge $e$ as input,
where $e$ is that edge among the edges $uw$ and $vw$ which lies in an unvisited face.
Note that the invariants of the Algorithm~\ref{algo:alg2} are maintained.
 
 The equivalence between line number 12 of Algorithm~\ref{algo:alg1} and 
 line number 16 of Algorithm~\ref{algo:alg2} follows from our arguments so far. 
 In line number 15 (respectively, 16) of Algorithm~\ref{algo:alg1}, a recursive call to the algorithm is made
 with input parameters: (1) $G_u$ (respectively, $G_v$), the $uw$-segment (respectively, $vw$-segment)
 of $G$ that does not contain the edge $vw$ (respectively, $uw$) and (2) the edge $uw$ (respectively, $vw$).
 Note that the edges bounding the face $uvw$, other than $uv$, are used as parameters in these two calls.
 Equivalently, in line number 17 and 18 of Algorithm~\ref{algo:alg2}, recursive calls are made respectively with 
 the edges $uw$ and $vw$ as input parameters.
 Since the face $uvw$ is marked as visited already, the invariants of the algorithm are maintained here as well.

 Thus, we can see that Algorithm~\ref{algo:alg2} is linear time
implementation of Algorithm~\ref{algo:alg1}.
\subsection{Building the Vertex Edge Face Adjacency List}
We can use the algorithm given below to populate the Vertex Edge Face Adjacency List of a maximal outerplanar graph.
\begin{algorithm}[H]
     \begin{algorithmic}[1]
     \Procedure {Populate Vertex Edge Face Adjacency List} {}
      \caption{For populating the Vertex Edge Face Adjacency List}
  \State \textbf{Inputs}: Bare vertex edge face adjacency list $L$ of a maximal outerplanar graph $G$, 
   in which face nodes are absent and links to face nodes from edge nodes are initialized to NULL.  
   \State \textbf{Outputs}: Modified vertex edge face adjacency list $L$ with all relevant details filled in.
    \label{algo:alg3}
     \State Create a copy $L'$ of $L$. 
      \State For any edge node $v_iv_j$ in $L'$, keep a pointer to and from the
            corresponding edge node in $L$. 
      \State Sort the vertices of $G$ in decreasing order of vertex degree. (using bucket sort in linear time)
       \For{\textit{i=1 to n}} 
           
             \State  Let $v_i=$ next highest degree vertex.
             \State  $p=s=Varray[v_i].head$ in $L'$. 
              \While {($s !=NULL$)} \Comment {Parse the edge list of $v_i$}.
                \State $x = s.vertex_2$
                \State $Varray[x].mark=Marked$
                \State $Varray[x].cur$-$edge= s$
                \State $s=s.next$
              \EndWhile
            
             \While {($p!=NULL$)}       
                 \State $v_j=p.vertex_2$
                 \State $q= Varray[v_j].head$ \Comment{Beginning of adjacency list of $v_j$}
                   \While {($q!=NULL$)}
                       \State $v_k=q.vertex_2$
                      \If {$(Varray[v_k].status=Marked)$}
                         \State $r=Varray[v_k].cur$-$edge$ 
                         \State $p'=$Copy$(p)$ in $L$, $q'=$Copy$(q)$ in $L$ and $r'=$Copy$(r)$ in $L$.
                         \State In $L$, create face node $f$ with links to edge nodes $p'$, $q'$ and \\ \hspace{3cm}$r'$ and reverse. 
                          \State Simultaneously insert links from $f$ to pair nodes of $p'$, $q'$ \\ \hspace{3cm}and $r'$ and reverse.
                      \EndIf
                      \State{$q=q.next$};
                    \EndWhile
                \State $v_j.mark=Unmarked$
                \State $p=p.next$;
             \EndWhile
           \State In $L'$, $Varray[v_i].mark=Deleted$ and delete the edge list of $v_i$. \\ 
            \hspace{0.8cm}Whenever edge $v_i,v_k$ is deleted, delete $v_k,v_i$ also.
         \EndFor
 
  \EndProcedure
 \end{algorithmic}
 \end{algorithm}

We assume a bare Vertex Edge Face Adjacency List $L$ as the input graph representation in which the face nodes and the links to and from face nodes are absent. 
The Algorithm~\ref{algo:alg2}, begins by creating a copy $L'$ of the bare Vertex Edge Face Adjacency List $L$ and links from an edge node in $L'$
to its copy in $L$. (For simplicity, in the data structure given in Fig.~\ref{fig:node},
we have avoided a field for this link.) During the running of Algorithm~\ref{algo:alg2}, we traverse through $L'$ and each time a new
face node is created corresponding to a face $f_i$ in $G$ , we assign links from the newly created face node to the edge nodes (bounding edges of $f_i$)
in $L$ and also the reverse links.
This step can be done in constant time by utilizing the links from each edge nodes of $L'$ to its copy in $L$. 
During the running of Algorithm~\ref{algo:alg2}, some edge nodes of $L'$ gets deleted while their copies in $L$ are retained. Note that, this way of
deletion of edge nodes of $L'$ is necessary to obtain linear time complexity, as it was done in \cite{chiba1985arboricity}.
Hence, the linear time complexity of the construction follows from the analysis in~\cite{chiba1985arboricity}. 

  \begin{figure}
 \centering
  \includegraphics[scale=.95]{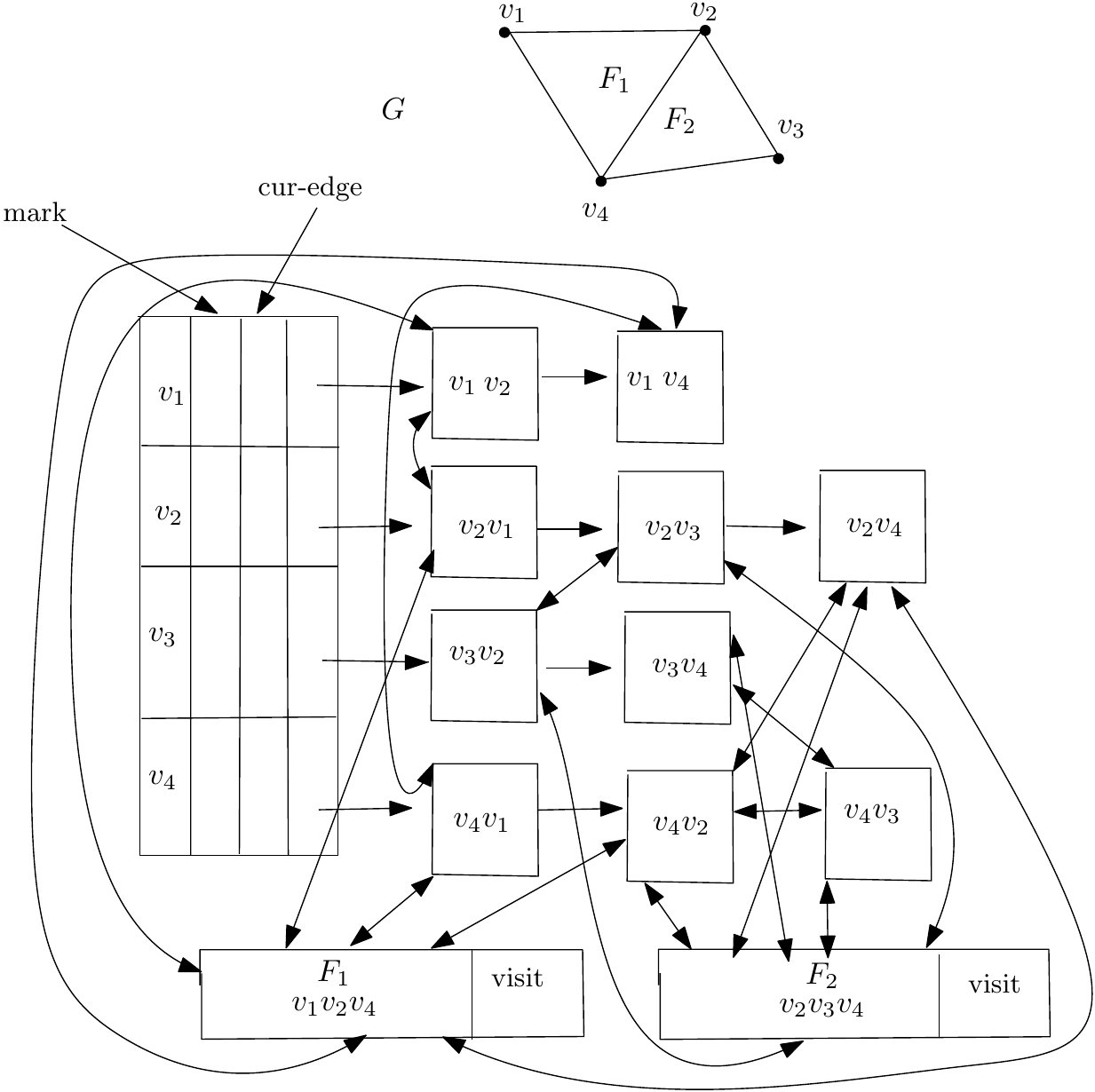}
  \caption{The Vertex Edge Face Adjacency List of a maximal outerplanar graph $G$ on four vertices and two faces. $F_1$ and $F_2$ denote the face nodes.}
   \label{fig:Example}
  \end{figure}
\end{document}